\documentclass[journal]{IEEEtran}
\usepackage{amsmath}
\usepackage{amsfonts}
\usepackage{amsthm}
\usepackage{amssymb}
\usepackage{algorithm}
\usepackage{array}
\usepackage{booktabs}
\usepackage{textcomp}
\usepackage{stfloats}
\usepackage{url}
\usepackage{verbatim}
\usepackage{graphicx}
\usepackage{balance}
\usepackage{xcolor}
\usepackage{color} 
\usepackage{nicefrac}
\usepackage{enumitem}
\usepackage{cases}
\usepackage{dsfont}
\usepackage[hidelinks]{hyperref}
\usepackage{romannum}
\usepackage{nomencl}
\usepackage{caption}
\usepackage{xpatch}
\usepackage{float}
\DeclareGraphicsExtensions{.pdf}

\newtheorem{assumption}{Assumption}
\newtheorem{proposition}{Proposition}
\newtheorem{lemma}{Lemma}

\theoremstyle{definition}
\newtheorem{remark}{Remark}
\newtheoremstyle{step}% name
  {\topsep}% space above
  {\topsep}% space below
  {}% body font
  {}% indent amount
  {\itshape}% theorem head font
  {.}% punctuation after theorem head
  {.3em}% space after theorem head
  {\thmname{#1}\thmnumber{ #2}\thmnote{ (#3)}}% theorem head spec
\theoremstyle{step}

\makenomenclature
\usepackage{etoolbox}
\renewcommand{\nomgroup}[1]{
\ifthenelse{\equal{#1}{A}}{\item[\textbf{Abbreviations and sets}]}{}
\ifthenelse{\equal{#1}{P}}{\item[\textbf{Parameters}]}{}
\ifthenelse{\equal{#1}{V}}{\item[\textbf{Variables}]}{}
}

\DeclareMathAlphabet{\mathbbb}{U}{bbold}{m}{n}

\newcommand{\col}{\operatorname{col}}
\newcommand{\diag}{\operatorname{diag}}
\newcommand{\proj}{\operatorname{proj}}
\newcommand{\argmin}{\operatorname{argmin}}
\newcommand{\N}{\operatorname{N}}
\newcommand{\xf}{x_{\rm tot}}
\newcolumntype{M}[1]{>{\centering\arraybackslash}m{#1}}
\usepackage{colortbl}
\makeatletter
\xpatchcmd{\proof}{\topsep6\p@\@plus6\p@\relax}{}{}{}
\makeatother

\begin{document}

\title{A Network-Constrained Demand Response Game for Procuring Energy Balancing Services}
\author{Xiupeng Chen, Koorosh Shomalzadeh, Jacquelien M. A. Scherpen, and Nima Monshizadeh
\thanks{Xiupeng Chen, Koorosh Shomalzadeh, Jacquelien M. A. Scherpen, and Nima Monshizadeh are with Jan C. Willems Center for Systems and Control, ENTEG, University of Groningen, Groningen, 9747 AG, The Netherlands. (email: {\tt \{xiupeng.chen, k.shomalzadeh, j.m.a.scherpen, n.monshizadeh\}@rug.nl})}
}

\maketitle

\setlength\abovedisplayskip{1.5pt}
\setlength\belowdisplayskip{1.5pt}
\setlength\intextsep{1pt}
\setlength{\abovecaptionskip}{0pt}

\begin{abstract}
Securely and efficiently procuring energy balancing services in distribution networks remains challenging, especially within a privacy-preserving environment. This paper proposes a network-constrained demand response game, i.e., a Generalized Nash Game (GNG), to incentivize energy consumers to offer balancing services. Specifically, we adopt a supply function-based bidding method for our demand response problem, where a requisite load adjustment must be met. To ensure the secure operation of distribution networks, we incorporate physical network constraints, including line capacity and bus voltage limits, into the game formulation. In addition, we analytically evaluate the efficiency loss of this game. Previous approaches to steer energy consumers toward the Generalized Nash Equilibrium (GNE) of the game often necessitated sharing some private information, which might not be practically feasible or desired. To overcome this limitation, we propose a decentralized market clearing algorithm with analytical convergence guarantees, which only requires the participants to share limited, non-sensitive information with others. Numerical analyses illustrate that the proposed market mechanism exhibits a low market efficiency loss. Moreover, these analyses highlight the critical role of integrating physical network constraints. Finally, we demonstrate the scalability of our proposed algorithm by conducting simulations on the IEEE 33-bus and 69-bus test systems.
\end{abstract}


\nomenclature[A,01]{\(\mathrm{GNG}\)}{Generalized Nash Game}
\nomenclature[A,02]{\(\mathrm{GNE}\)}{Generalized Nash Equilibrium}
\nomenclature[A,03]{\(\text{v-GNE}\)}{variational Generalized Nash Equilibrium}
%\nomenclature[A,04]{\(\text{NE}\)}{Nash Equilibrium}
\nomenclature[A,05]{\(\mathrm{DSO}\)}{Distribution System Operator}
%\nomenclature[A,06]{\(\text{BRP}\)}{Balance Responsible Party}
\nomenclature[A,07]{\(\mathrm{PoA}\)}{Price of Anarchy}
\nomenclature[A,07]{\(\mathrm{LI}\)}{Lerner Index}
\nomenclature[A,07]{\(\mathrm{DWL}\)}{Deadweight Loss}
\nomenclature[A,08]{\(\mathcal{N},\mathcal{M}\)}{sets of active consumers and passive consumers }
%\nomenclature[A,09]{\(\mathcal{M}\)}{set of passive consumers}
\nomenclature[A,10]{\(\mathcal{B},\mathcal{L}\)}{sets of buses and lines}
\nomenclature[A,11]{\(\mathcal{B}_{\rm out}^b\)}{set of out-neighbors of bus $b$}
\nomenclature[A,12]{\(\mathcal{B}_{\rm in}^b\)}{set of in-neighbors of bus $b$}
\nomenclature[A,13]{\(\mathcal{M}^b\)}{set of passive consumers connected at bus $b$}
\nomenclature[A,14]{\(\mathcal{N}^b\)}{set of active consumers connected at bus $b$}
%\nomenclature[A,15]{\(\mathcal{L}\)}{set of lines}
%\nomenclature[A,16]{\(\mathbb{R},\mathbb{R}^+\)}{real number set, nonnegative real number set}
%\nomenclature[A,17]{\(\N(\cdot)\)}{normal cone operator}
%\nomenclature[A,18]{\(\proj(\cdot)\)}{projection operator}
%\nomenclature[A,19]{\(\col(\cdot)\)}{stacked vector obtained from $\cdot$}
%\nomenclature[A,20]{\(\diag(\cdot)\)}{diagonal matrix with $\cdot$ on its diagonal}
%\nomenclature[A,21]{\(\lambda_{\min}(\cdot)\)}{minimum eigenvalue of a symmetric matrix $\cdot$}
%\nomenclature[A,22]{\(\lambda_{\max}(\cdot)\)}{maximum eigenvalue of a symmetric matrix $\cdot$}

\nomenclature[P,01]{\(N,M\)}{numbers of active consumers and passive consumers}
%\nomenclature[P,02]{\(M\)}{number of passive consumers}
\nomenclature[P,03]{\(B,L\)}{numbers of buses and lines}
%\nomenclature[P,04]{\(L\)}{number of lines}
%\nomenclature[P,05]{\(I\)}{identity matrix}
%\nomenclature[P,06]{\(\mathbbb{1}(\mathbbb{0})\)}{vector with all elements equal to 1 (0)}
\nomenclature[P,07]{\(\alpha,\kappa\)}{price sensitivity, Lipschitz constant}
\nomenclature[P,08]{\(\xf\)}{load adjustment requirement}
\nomenclature[P,09]{\(\hat{x}_n\)}{maximum available flexibility of active consumer $n$}
%\nomenclature[P,10]{\(\kappa\)}{Lipschitz constant}
\nomenclature[P,11]{\(d_n\)}{pre-scheduled net load of consumer $n$}
\nomenclature[P,12]{\(u_{(b,s)}\)}{conductance of line $(b,s)$}
\nomenclature[P,13]{\(w_{(b,s)}\)}{susceptance of line $(b,s)$}
\nomenclature[P,14]{\(z_{(b,s)}\)}{maximum capacity of line $(b,s)$}
\nomenclature[P,15]{\(\underline{\theta},\overline{\theta}\)}{bounds of the voltage phase angles}
\nomenclature[P,16]{\(\underline{v},\overline{v}\)}{bounds of the voltage phase magnitudes}

\nomenclature[V,01]{\(\beta_n\)}{bid of active consumer $n$}
\nomenclature[V,02]{\(x_n\)}{flexibility of active consumer $n$}
\nomenclature[V,03]{\(\lambda\)}{uniform market clearing price}
\nomenclature[V,04]{\(p_b,q_b\)}{active and reactive power injections at bus $b$}
\nomenclature[V,05]{\(v_b,\theta_b\)}{voltage magnitude and angle of bus $b$}
\nomenclature[V,06]{\(p_{(b,s)}\)}{active power flow through line $(b,s)$}
\nomenclature[V,07]{\(q_{(b,s)}\)}{reactive power flow through line $(b,s)$}
\printnomenclature
%{\footnotesize
%\printnomenclature[0.5in]
%}

\section{Introduction}
\IEEEPARstart{T}{he} increased integration of renewable energies brought an unprecedented challenge to cost-effectively balancing energy supply and demand \cite{sinsel2020challenges,silva2022short}. 
%Thanks to the advanced communication and metering technologies, energy consumers could provide energy balancing services by modifying their energy consumption in response to energy supply fluctuations. System operators procure such balancing services through various demand response programs \cite{siano2014demand}.
Advanced communication and metering technologies enable consumers to modify energy consumption in response to supply fluctuations, providing balancing services through demand response programs \cite{siano2014demand}. Due to scalability issues, these consumers typically participate in energy markets via entities, such as utility companies, aggregators, virtual power plants, etc. \cite{gkatzikis2013role, naval2021virtual}. %However, these entities usually manage their energy consumers using centralized mechanisms to achieve the requisite balancing services \cite{correa2019optimal, sadeghi2018optimal, jiang2022flexibility}. These centralized mechanisms curtail the incentives for consumers to offer balancing services, as they assume consumers to be mere price-takers, despite their potential to make strategic decisions for financial benefits \cite{bahramara2017modeling}.
These entities usually use centralized mechanisms to achieve the requisite balancing services \cite{correa2019optimal, sadeghi2018optimal, jiang2022flexibility}, which reduce incentives for consumers to offer balancing services by treating them as price-takers, despite their potential for strategic decision-making \cite{bahramara2017modeling}. Additionally, such mechanisms require entities to access market participants' private information to determine market equilibria. In light of these issues, %this paper proposes a demand response game for a utility company to procure balancing services from energy consumers in a decentralized manner.  This method accounts for the self-interested behaviors of consumers while protecting their private information. 
this paper proposes a decentralized demand response game for utility companies to procure balancing services, considering consumers' self-interested behaviors and protecting their private information.

Game theory, both in the form of cooperative and noncooperative games, 
%has been utilized as a powerful tool to analyze the strategic demand response behaviors of energy consumers \cite{saad2012game, deng2015survey}.
is a powerful tool for analyzing strategic demand response behaviors of energy consumers \cite{saad2012game, deng2015survey}. In cooperative games, participants establish agreements to attain a specific common goal, ensuring that the benefits derived from the collaboration are equitably distributed \cite{cui2015residential,sanjab2022tso}. %Yet, the enforcement of these agreements presents challenges, and solving these problems necessitates access to private information. 
However, enforcing these agreements is challenging and requires access to private information. In Stackelberg games, energy consumers remain price-takers, adjusting their loads in response to the price signals from utility companies \cite{aguiar2021network,pandey2021hierarchical}. The (generalized) Nash game is able to model the competitive behaviors among price-makers, addressing day-ahead energy scheduling problems through various billing methods \cite{mishra2022game,chen2014autonomous,hupez2022pricing}. %For instance, \cite{mishra2022game,chen2014autonomous, hupez2022pricing} formulate demand response games based on varied billing methods to address day-ahead energy scheduling problems. 
% In these studies \cite{aguiar2021network}-\cite{hupez2022pricing}, energy consumers adopt a quantity-based bidding method, with bids representing the quantities of energy consumption. This bidding method is particularly suitable for the demand response problems in \cite{aguiar2021network}-\cite{hupez2022pricing}, where energy consumers determine their consumption first, followed by the energy providers meeting the total demand. 
In these studies \cite{aguiar2021network}-\cite{hupez2022pricing}, energy consumers use a quantity-based bidding method, with bids representing energy consumption quantities. This method suits demand response problems where consumers determine their consumption first, followed by energy providers meeting the total demand. However, this bidding method is challenging for problems requiring matching a predetermined energy supply, necessitating careful selection of billing parameters to ensure demand meets supply \cite{ma2014distributed}. %However, this bidding method proves challenging to implement for the demand response problems with the additional constraint that a predetermined energy supply must be matched. For instance, billing parameters in the energy consumption game need a careful selection to ensure that demand meets supply \cite{ma2014distributed}.

The supply function bidding method is more flexible for demand response problems as it allows market players to adapt their quantities to varying prices, rather than committing to a fixed price or quantity \cite{klemperer1989supply}. While this method is predominantly used in the wholesale market \cite{hobbs2000strategic}, only a few pieces of literature employ it to depict the bidding behaviors of energy consumers in distribution networks. The most relevant works to ours \cite{li2015demand, xu2015demand} address the energy balancing problems by suggesting two distinct forms of supply function bids. The studies \cite{chen2019energy,chen2022energy} employ supply function-based bidding methods to handle energy-sharing problems. Although these market mechanisms \cite{li2015demand}-\cite{chen2022energy} present considerable advantages in terms of economic efficiency enhancements, they neglect the physical feasibility of the market equilibriums. Specifically, \cite{li2015demand, chen2019energy} do not take into account either local capacity limitations or physical network constraints, while the work \cite{xu2015demand} considers only local capacity constraints  \cite{xu2015demand}. The study \cite{chen2022energy} includes simplified DC power flow constraints but excludes local capacity constraints and bus voltage limits in distribution networks. In this paper, on the other hand, we incorporate both local capacity and physical network constraints, including line capacity and voltage limits, in an AC power flow model. Due to the aforementioned practical considerations, the clearing methods in \cite{li2015demand}-\cite{chen2022energy} are not readily applicable to our setup. Moreover, the presence of local capacity and physical network constraints substantially adds to the complexity of the corresponding game formulation and efficiency analysis.

We model the competition among energy consumers as a Generalized Nash Game (GNG), where the bid of each consumer is intertwined with complex physical network constraints. The market equilibrium is the Generalized Nash Equilibrium (GNE) where each consumer maximizes its profit, and none have an incentive to deviate unilaterally \cite{kulkarni2012variational}. The distributed iterative algorithms for equilibrium-seeking presented in \cite{mishra2022game}-\cite{hupez2022pricing} are not suitable for our problem, given their lack of consideration for coupling constraints. The Nikaido-Isoda function-based  algorithm is employed in \cite{li2022data}, which necessitates a central coordinator to gather potentially privacy-sensitive information from all participants. The energy sharing game is turned into an equivalent optimization problem to indirectly determine the GNE in \cite{wang2021distributed}. 
In addition, recent progress in the control systems community includes the development of advanced decentralized algorithms for computing GNE \cite{pavel2019distributed, yi2022survey}. Nonetheless, directly implementing these algorithms in our setup is challenging, mainly because they are best suited for problems with linear coupling constraints \cite{scarabaggio2021distributed,atzeni2014noncooperative}. 
For AC power flow models with nonlinear security constraints, the algorithm in \cite{belgioioso2022operationally} demands a specific structure for the objective functions of energy consumers, while the algorithm introduced in \cite{scarabaggio2022noncooperative} requires the coordinator to manage all coupling constraints. However, these coupling constraints might encompass sensitive data, such as local capacity constraints in our scenario. Given these algorithmic limitations, we develop a decentralized market clearing method for our game based on a preconditioned forward-backward splitting technique \cite{belgioioso2018projected}.

In this paper, we consider a demand response problem where the utility company procures energy balancing services from energy consumers to fulfill a load adjustment requirement. These active energy consumers are equipped with flexible resources and could modify their pre-scheduled energy consumption or generation in the day-ahead markets. 
On the one hand, these consumers behave rationally and bid strategically to maximize their profits; therefore, studying the competition among consumers and their optimal bidding strategies is crucial. On the other hand, energy consumers are physically connected within the distribution network. Thus, incorporating realistic physical network constraints is significant to ensure the secure operation of the distribution network.
In this regard, we propose a network-constrained demand response game for procuring balancing services in distribution networks, where energy consumers adopt the supply function-based bidding method. Then, we develop a decentralized market clearing method to steer consumers towards the market equilibrium. The main contributions of this work are as follows:
\begin{itemize}
    \item \textit{Game design with AC power flow:} The proposed demand response game integrates both local capacity constraints and AC power flow network constraints, including line capacity and bus voltage limits. In contrast to the supply function-based bidding methods in \cite{li2015demand}-\cite{chen2022energy}, the proposed method ensures the secure operation of the distribution network. Numerical results highlight the significance of incorporating such constraints in the design.
    \item \textit{Game equilibrium analysis:} Considering both local capacity and physical network constraints, we formulate the demand response game among energy consumers as a Generalized Nash Game (GNG). We characterize the market equilibrium, referred to as the variational Generalized Nash Equilibrium (v-GNE), which represents the best bidding strategies of the energy consumers in the demand response game. We prove the existence and uniqueness of this equilibrium. Moreover, the efficiency loss associated with the formulated game has been analytically evaluated and an upper bound is established. 
    \item \textit{Privacy-aware market clearance:} Our work introduces a decentralized market clearing algorithm designed for the demand response game, emphasizing privacy preservation amidst information sharing constraints. Compared to the centralized methods \cite{correa2019optimal, sadeghi2018optimal, jiang2022flexibility}, the proposed  algorithm ensures the confidentiality of private data, allowing only the exchange of non-sensitive information between market participants, even in the presence of complex and nonlinear coupling security constraints. We have analytically proven the algorithm's convergence and validated its scalability with a comprehensive case study. 
\end{itemize}

This paper is structured as follows. In Section \Romannum{2}, we present the system model for procuring energy balancing services in a distribution network. In Section \Romannum{3}, we formulate the competition among energy consumers as a game and evaluate the efficiency of its equilibrium. Section \Romannum{4} introduces the algorithm for computing the v-GNE in the game. Section~\Romannum{5} provides a comprehensive case study to demonstrate the effectiveness of the proposed market mechanism and clearing method. The paper closes with conclusions in Section \Romannum{6}.

\textit{Notation and preliminaries:} 
Let $\mathbb{R}$ and $\mathbb{R}^+$ be the sets of real numbers and nonnegtive real numbers. We use $\mathbbb{1}$($\mathbbb{0}$) to denote the vector/matrix with all elements equal to 1(0) and use $I$ as the identity matrix. Given a set $\mathcal{N}=\{1,2,...,N\}$, $\col(x_n)_{n\in\mathcal{N}}$ ($\diag(x_n)_{n\in\mathcal{N}}$) denotes the stacked vector (diagonal matrix) obtained from $x_n$. For a matrix $A\in \mathbb{R}^{N\times N}$, we use $\lambda_{\min}(A)$ ($\lambda_{\max}(A)$) to denote its minimum (maximum) eigenvalue. Given a closed convex set $\Omega\in\mathbb{R}_n$, we denote the projection of a point $x\in\mathbb{R}_n$ to the set $\Omega$ by ${\proj_{\Omega}(x)=\argmin_{y\in\Omega}\|y-x\|}$. We denote the normal cone operator by $\N_{\Omega}(x)$, i.e, $\N_{\Omega}(x)=\emptyset$ if $x\notin\Omega$, and $\N_{\Omega}(x)=\{v\in \mathbb{R}^n\mid \mathrm{sup}_{z\in\Omega}\langle v,z-x\rangle\leq 0\}$ otherwise. Given an operator $F:\mathcal{X}\rightarrow \mathbb{R}^n$, the variational inequality problem VI($\mathcal{X}$,F) is to find the point $\bar{x}\in \mathcal{X}$ such that $(x-\bar{x})^\top F(\bar{x})\geq 0,\,\forall \, x\in \mathcal{X}$. The operator $F$ is $\xi$-averaged, with $\xi\in(0,1)$, if $\|F(x)-F(y)\|^2\leq\|x-y\|^2-\frac{1-\xi}{\xi}\|(x-F(x))-(y-F(y))\|^2$. 

\section{System model}
\subsection{Market model}\label{balancing market model}
This paper considers a distribution network consisting of a utility company and two types of energy consumers, namely active and passive ones denoted by $\mathcal{N}=\{1,2,\dots,N\}$ and $\mathcal{M}=\{N+1,N+2,\dots,N+M\}$, respectively. An active consumer is equipped with at least one flexible resource, such as a dispatchable generator or an interruptable load, while passive consumers have fixed loads and no dispatchable generation units. In addition, the Distribution System Operator (DSO) ensures the secure operation of the distribution network. Fig.~\ref{fig:relationship} illustrates the physical and communication connection among the market participants in our demand response problem.

The utility company supplies energy for these energy consumers. Moreover, it can incentivize active consumers to adjust their energy consumption or generation to meet a load adjustment requirement $\xf\in\mathbb{R}^{+}$. The latter means the total amount of energy consumption scheduled in the day-ahead market should be reduced or increased by $\xf$ for real-time energy balancing. 
The value of $\xf$ can be determined in the upstream network via a price-based demand response or incentive-based demand response programs \cite{khajavi2011role}. Alternatively, $\xf$ can be a result of  prediction errors associated with renewable generation \cite{li2015demand}. In any case, the subsequent developments are independent of how $\xf$ is determined as long as the utility company is informed about its value. The goal of the utility company is to determine the allocation of flexibility $x_n$ (i.e., the extent of load adjustment) that each active consumer $n\in\mathcal{N}$ should provide to meet the total load adjustment requirement, that is
 \begin{equation}\label{eq_flexibility_balance}
    \sum_{n\in\mathcal{N}} x_n=\xf.
 \end{equation}
Note that the above demand response problem is different from the ones in \cite{mishra2022game}-\cite{hupez2022pricing}. These energy scheduling problems only consider how much energy the active consumers can consume or generate, while here we solve a demand response problem with the extra constraint \eqref{eq_flexibility_balance}, which makes the quantity-based bidding methods difficult to implement \cite{ma2014distributed}.
\begin{figure}[ht]
\begin{center}
\includegraphics[width=0.42\textwidth]{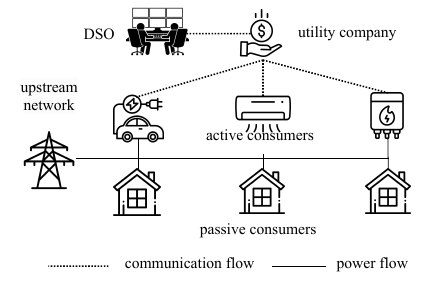}
\end{center}
%\vspace*{-2mm}
\caption{The connection among market participants}
\label{fig:relationship}
\end{figure}  

Motivated by the scheme proposed in \cite{chen2019energy} to solve an energy sharing problem, we consider the bidding method and pricing mechanism for our demand response scheme as:

\begin{itemize}
\item Each active consumer submits its strategic decision variable $\beta_n$ to the utility company.

\item The utility company determines the flexibility price $\lambda$ that it pays to active consumers as
\begin{equation}\label{eq_price}
    \lambda=\frac{\xf-\mathbbb{1}^{\top}\beta}{\alpha N},
\end{equation}
and the corresponding allocated flexibility for active consumer $n$ as
\begin{equation}\label{eq_flexibility_allocation}
    x_n=\frac{\xf-\mathbbb{1}^{\top}\beta}{N}+\beta_n,
\end{equation}
where $\beta=\col(\beta_n)_{n\in\mathcal{N}}$, and $\alpha \in\mathbb{R}^+$ is a constant imposed by the utility company. 
Note that under the above scheme, the constraint \eqref{eq_flexibility_balance} always holds for any $\beta\in\mathbb{R}^N$.

\item The DSO is responsible for verifying that the allocated flexibilities $x_n$s meet the AC power flow network constraints as detailed in the next subsection. The active consumers' bids should be modified if the physical network constraints are not met. The modification method is provided in Section \ref{algorithm}.

\end{itemize}
\begin{remark} 
The above bidding method is essentially based on the supply function bidding method. Each active consumer $n \in \mathcal{N}$ submits a linear supply function in the form of $ x_n=\alpha \lambda+\beta_n$. The variable $\beta_n$ indicates consumer $n$'s willingness to provide flexibility, while the parameter $\alpha$ represents the overall price sensitivity of the market (for further details and explanations refer to \cite{chen2019energy}). Subsequently, after bid submission, the utility company clears the price $\lambda$ such that \eqref{eq_flexibility_balance} holds. There are other forms for a bidding function, such as $x_n=\beta_n \lambda$ in \cite{li2015demand} and $x_n=\hat{x}_n-\beta_n/\lambda$ in \cite{xu2015demand}, where $\hat{x}_n$ is the maximum flexibility capacity of consumer $n$. We compare the market efficiency performances of different bidding functions in Section~V. Furthermore, in contrast to the market models in \cite{li2015demand,xu2015demand,chen2019energy}, we incorporate the bid validation and modification phase conducted by the DSO under the AC power flow model. This integration necessitates distinct approaches to game formulation, efficiency analysis and decentralized algorithm design.
\end{remark}
\subsection{Physical network model}\label{subsec:network}
In this subsection, we elaborate on the AC power flow network constraints relevant to the bid validation and modification phase. First, we present the mathematical model of the distribution network, where all consumers are physically connected.
We represent the distribution network $\mathcal{D}(\mathcal{B},\mathcal{L})$ as a simple directed graph, where the buses and lines are represented by $\mathcal{B}$ and $\mathcal{L}\subseteq \mathcal{B} \times \mathcal{B}$, respectively. 
That is, if $\ell=(b,s) \in \mathcal{L}$, then line $\ell$ originates from bus $b\in \mathcal{B}$  and ends in bus $s \in \mathcal{B}$. The numbers of buses and lines are $B$ and $L$. Note that the direction of lines can be chosen arbitrarily. We denote the set of out-neighbors and in-neighbors of $b$ by 
\[
\mathcal{B}_{\rm out}^b= \{s \mid (b,s) \in \mathcal{L}\}, \quad \mathcal{B}_{\rm in}^b= \{s \mid (s,b) \in \mathcal{L}\},
\]
respectively. Furthermore, we define the incidence matrix $E=[e_{\ell,b}]\in\mathbb{R}^{L\times B}$ as 
\begin{equation}
    e_{\ell,b}=
    \begin{cases}
        +1 \quad  &\text{if line $\ell$ leaves bus $b$,} \\
        -1 \quad   &\text{if line $\ell$ enters bus $b$,} \\
        0 \quad &\text{otherwise}.
    \end{cases}
\end{equation}
%\vspace{.3cm}
\begin{figure}[ht!]
\begin{center}
\includegraphics[width=0.45\textwidth]{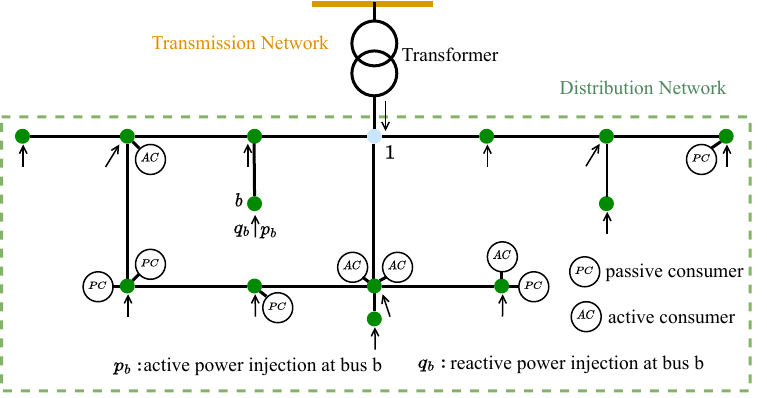}
\end{center}
\caption{An example of physical network}
\label{fig:physical}
\end{figure}
%\vspace{.2cm}

The transmission and distribution networks are connected via bus $1$. Besides, we use $\mathcal{M}^b$ and $\mathcal{N}^b$ to denote the set of passive and active consumers connected to each bus $b\in \mathcal{B}\setminus \{1\}$. Fig. \ref{fig:physical} depicts an example of distribution network.

By adopting the linear lossless power flow equations in \cite{yang2016state}, the active and reactive power flow $p_{(b,s)}$ and $q_{(b,s)}$   through line $(b,s)\in\mathcal{L}$ are given by
\begin{equation}\label{eq_line_balance}
    \begin{aligned}
      & p_{(b,s)} = -w_{(b,s)} (\theta_b-\theta_s) + u_{(b,s)} (v_b - v_s), \\
      & q_{(b,s)} = -u_{(b,s)} (\theta_b-\theta_s) - w_{(b,s)} (v_b - v_s), 
    \end{aligned}    
\end{equation}
where $u_{(b,s)}$ and $w_{(b,s)}$ denote the conductance and susceptance of line $(b,s)$, and $v_b$ and $\theta_b$ indicate the voltage magnitude and angle of bus $b\in \mathcal{B}$, respectively. 

For each bus $b\in\mathcal{B}$, the power balance must hold, namely
\begin{equation}\label{eq_bus_balance}
    \begin{aligned}
      &  \sum_{s\in \mathcal{B}^b_{ \rm out}} p_{(b,s)}- \sum_{s\in \mathcal{B}^b_{ \rm in}} p_{(s,b)}  = p_b, \\
      & \sum_{s\in \mathcal{B}^b_{ \rm out}} q_{(b,s)}- \sum_{s\in \mathcal{B}^b_{ \rm in}} q_{(s,b)}  = q_b.
    \end{aligned} 
\end{equation} 
where $p_b$ and $q_b$ are the net active and reactive power injections at bus $b \in \mathcal{B}$.
Note that bus $1$ is considered the slack bus in the power flow equations above. Consequently, we have $\theta_1 = 0$ and $v_1$ is a parameter indicating the output voltage of the on-load tap changer (OLTC) transformer as a discrete Volt/VAR device. The value of $v_1$ is determined by the DSO through adjusting the transformer winding ratios. 
For the other buses, the reactive power injection is given by $q_b = - \sum_{n\in\mathcal{M}^b\cup \mathcal{N}^b } q_n + q^r_b $, where $q_n$ is pre-scheduled reactive power of the consumer $n$ in the day-head market and $q^r_b$ is the output of the reactive power compensator as  a continuous Volt/VAR device connected at bus $b$. In case a bus is not connected to a compensator, the value of $q^r_b$ is set to zero.

%For the buses with reactive power compensators as continuous Volt/VAR devices, the reactive power injection is given $q_b = - \sum_{n\in\mathcal{M}^b\cup \mathcal{N}^b } q_n + q^r_b $, where $q_n$ is pre-scheduled reactive power of the consumer $n$ in the day-head market and  $q^r_b$ is the output of the reactive power compensator connected at bus $b$.}
Next, we impose the security constraints  for each bus $b \in \mathcal{B}$ and line $(b,s)\in \mathcal{L}$ as \cite{belgioioso2022operationally}
\begin{equation}\label{eq_physical_constraint}
\begin{aligned}
  & p_{(b,s)}^2 + q_{(b,s)}^2 \leq  z_{(b,s)}^2,\\
  & \underline{\theta} \leq \theta_b \leq \overline{\theta},\\
  & \underline{v} \leq v_b \leq \overline{v},
\end{aligned}   
\end{equation}
where the first inequality represents the line capacity constraint at each line, with the maximum line capacity denoted by $z_{(b,s)}$, and the last two inequalities represent the safe limits of the voltage phase angles and magnitudes. 
\par
We write the power flow and security constraints \eqref{eq_line_balance}, \eqref{eq_bus_balance} and \eqref{eq_physical_constraint} as follows.
\begin{equation}
    \begin{aligned}\label{eq_pf_compact}
        & P_{\mathcal{L}}  = -W E \theta + U E v,  \\
    & Q_{\mathcal{L}}  = -U E \theta -W E v, \\
    &E^\top P_{\mathcal{L}} = P_{\mathcal{B}}, \quad  E^\top Q_{\mathcal{L}} = Q_{\mathcal{B}}, \\
      & P_{\mathcal{L}} \odot P_{\mathcal{L}}  + Q_{\mathcal{L}} \odot Q_{\mathcal{L}} \leq  Z \odot Z,\\
  & \underline{\theta} \mathbbb{1} \leq \theta \leq \overline{\theta} \mathbbb{1},\quad 
  \underline{v}\mathbbb{1} \leq v \leq \overline{v}\mathbbb{1},
    \end{aligned}
\end{equation}
where $\odot$ is the Hadamard product, $E$ is the incidence matrix of the distribution network and 
\begin{equation}
\begin{aligned}
&P_{\mathcal{B}}=\col(p_b), \ Q_{\mathcal{B}}=\col(q_b), \\
&P_{\mathcal{L}}=\col(p_{(b,s)}), \ Q_{\mathcal{L}}=\col(q_{(b,s)}), \\
&W=\diag(w_{(b,s)}), \ U=\diag(u_{(b,s)}), \\
&Z=\col(z_{(b,s)}), \ \theta=\col(\theta_b), \ v=\col(v_b).
\end{aligned}
\end{equation}

Finally, the flexibility provision is incorporated by calculating the net active power injections as
\begin{equation}\label{eq_bus_2}
        p_b = - \big( \sum_{n\in\mathcal{M}^b} d_n + \sum_{n\in\mathcal{N}^b} (d_n \mp x_n)\big),
\end{equation}
where $d_n$ is pre-scheduled net load of the consumer $n \in \mathcal{N}\cup \mathcal{M}$ in the day-head market. The flexibility $x_n$ appears with a minus (plus) sign if there is a supply deficit (surplus) in the distribution network, leading the consumers to inject (withdraw) additional power to/from the network. Throughout the rest of this paper, we indicate the DSO physical network constraints by   \eqref{eq_pf_compact} and \eqref{eq_bus_2}.
\begin{remark}\label{remark:powermodel}
   The accurate AC power flow model, due to its nonconvexity and nonlinearity, brings tractability and computational challenges to the proposed method. Therefore, we have employed the above linear power flow model which considers reactive power and bus voltage and provides fairly accurate and robust results \cite{yang2016state,yang2018general}. We note that the subsequent analysis and results hold for nonlinear, yet convex AC power flow models, such as the Second-Order Cone Relaxed DistFlow (SOCRD) model \cite{farivar2013branch} for a single-phase equivalent network and the Semidefinite Program (SDP) relaxed model \cite{lu2023risk} for an unbalanced
three-phase network. As a practical guideline, the bounds in \eqref{eq_physical_constraint} can be slightly restricted to accommodate for any modelling errors resulting from approximating the accurate power flow. 
\end{remark}
\section{Game formulation and analysis}
This section investigates the bidding behaviors of active energy consumers participating in the market mentioned in Section \Romannum{2}. We consider each active consumer in the distribution network to be rational, making strategic decisions in order to provide energy balancing services. First, we formulate the competition among these consumers as a game and then characterize the resulting (generalized) Nash equilibrium. At this equilibrium,  no consumer is
willing to unilaterally deviate from its bid. Additionally, this equilibrium guarantees the feasibility of both local flexibility capacity and physical network constraints. 
The local capacity constraint ensures each consumer $n\in\mathcal{N}$ can provide its allocated flexibility $x_n$, while the physical network constraints ensure the secure operation of the distribution network.

\subsection{Game formulation}
From the perspective of active consumer $n$, its ultimate aim of submitting bid $\beta_n$ is to maximize its net revenue, that is, to minimize
   \begin{equation} \label{eq_objective_original}
       J_n(x_n)=C_n(x_n)-\lambda x_n,
   \end{equation}
with the following local flexibility capacity constraints
\begin{equation} \label{eq_constraint_original}
        0 \le x_n \le \hat{x}_n.
    \end{equation} 
The function $C_n(x_n)$ indicates the cost or disutility of providing flexibility and the term $\lambda x_n$ is the payment from the utility company. Furthermore,  $\hat{x}_n$ is consumer $n$'s maximum available flexibility. Throughout the paper, we make the following assumption on $C_n(x_n)$.

\begin{assumption}\label{assumption}
The cost function $C_n(x_n)$ is convex, twice continuously differentiable, $C_n(0)\geq 0$ and $C_n(x_n)>0$ for  $x_n>0$. Moreover, $C_n^\prime=\nabla C_n$ is Lipschitz continuous with $\kappa>0$, i.e.,
\begin{equation}\label{eq_lipschitz}
    \lVert C'_n(x'_n)-C'_n(x_n)\rVert \leq \kappa \lVert x'_n-x_n\rVert, \forall \, 0 \le x_n, x'_n \le \hat{x}_n.
\end{equation}
\end{assumption}

Following the market operation proposed in Section \ref{balancing market model}, one can substitute $\lambda$ and $x_n$ from \eqref{eq_price} and \eqref{eq_flexibility_allocation} in \eqref{eq_objective_original} and rewrite it as
\begin{multline}\label{eq_objective}
\bar J_n(\beta_n,\beta_{-n})=C_n\big((\xf-\mathbbb{1}^{\top}\beta)/N+\beta_n\big) \\
-(\xf-\mathbbb{1}^{\top}\beta+N\beta_n)(\xf-\mathbbb{1}^{\top}\beta)/(\alpha N^2),   
\end{multline}
where $\beta_{-n}=\col(\beta_m)_{m\in\mathcal{N} \setminus \{n\}}$.
Similarly, we can rewrite  the constraint in \eqref{eq_constraint_original} as
\begin{equation}\label{eq_global_set1}
(\xf-\mathbbb{1}^{\top}\beta)/N+\beta_n \leq \hat{x}_n,
\end{equation}
\begin{equation}\label{eq_global_set2}
(\xf-\mathbbb{1}^{\top}\beta)/N+\beta_n \geq 0,
\end{equation}
for all $n\in \mathcal{N}$,
and the one in \eqref{eq_bus_2} as
\begin{equation}\label{eq_global_set3}
        p_b  = -\sum_{n\in\mathcal{N}^b_p\cup \mathcal{N}^b_a} d_n \pm
        \sum_{n\in\mathcal{N}^b} \big((\xf-\mathbbb{1}^{\top}\beta)/N+\beta_n\big),
\end{equation}
for all $b\in\mathcal{B}\setminus\{1\}$.

As a result, each active consumer $n \in \mathcal{N}$ aims to minimize 
$\bar J_n(\beta_n, \beta_{-n})$ subject to the flexibility capacity constraints \eqref{eq_global_set1} and \eqref{eq_global_set2}. 
Augmenting the latter optimization with the security constraints \eqref{eq_pf_compact} and \eqref{eq_global_set3} enforced by the DSO, we arrive at the following optimization problem  
\begin{equation} \label{prob_consumer}
    \begin{aligned}
        &\min_{\beta_n} \quad &&\bar J_n(\beta_n, \beta_{-n}) \\
        &\mathrm{s.t.} && \beta_n \in K_n(\beta_{-n}),
    \end{aligned}
\end{equation}
where the parametric set $K_n(\beta_{-n})$ is defined as
\[
    \{ \beta_n \in \mathbb{R} \mid \eqref{eq_global_set1},\eqref{eq_global_set2},
        \eqref{eq_global_set3} \ \mathrm{and} \ \eqref{eq_pf_compact} \, \text{hold for some} \, v\,\text{and} \, \theta \}.
\]
From \eqref{prob_consumer}, we note that both the objective function and constraints of each active consumer depend on its own strategy as well as the strategies of other consumers.

We now write the noncooperative game among active consumers in a compact form as the triple: 
\begin{equation}\label{eq_game}
\mathcal{G} =\{\mathcal{N},K,\col(\bar J_n(\beta_n,\beta_{-n}))_{n\in \mathcal{N}}\}, 
\end{equation}
where $K=\prod\limits_{n\in \mathcal{N}}K_n(\beta_{-n})$ is the set of feasible strategies for the consumers.  

Noting that convexity is preserved under affine transformations, it is easy to verify that the set $K$ is convex. Furthermore, to satisfy Slater's constraint qualifications \cite{boyd2004convex}, we assume that $K$ has at least one strictly feasible point.

The competition among consumers in the game $\mathcal{G}$ gives rise to a GNG since their objective functions and the feasible strategy sets are coupled. Next, we analyze the GNE of this game. A point $\beta^* \in K$ is a  GNE of the game, if for all $n\in\mathcal{N}$, the following holds,
\begin{equation}
 \bar J_n(\beta_n,\beta^*_{-n})\geq \bar J_n(\beta^*_n,\beta^*_{-n}),\ \ \forall \  \beta_n \in K_n(\beta^*_{-n}).
\end{equation}

At this point each consumer can minimize its objective function given the bidding strategies of others. Therefore, none of the consumers would unilaterally deviate from its strategy. In this paper, we focus on a specific subclass of GNE, namely v-GNE. Specifically, each player in the game is penalized equally for deviating from coupling constraints at the v-GNE, which corresponds to the solution of a variational inequality problem \cite{kulkarni2012variational}. To characterize the v-GNE of the game $\mathcal{G}$ and verify its existence and uniqueness, we define the
pseudo-gradient  mapping of the game as 
\begin{equation} \label{eq_pgm}
    F=\col(f_n(\beta_n,\beta_{-n}))_{n \in \mathcal{N}},
\end{equation}
where 
\begin{equation}
\begin{aligned}
f_n(\beta_n,\beta_{-n})&=\frac{\partial}{\partial\beta_n}\bar J_n(\beta_n,\beta_{-n}) \\=C'_n(x_n)&\frac{N-1}{N}
+\frac{(\mathbbb{1}^{\top}\beta-\xf)(N-2)+N \beta_n}{\alpha N^2}.
\end{aligned}    
\end{equation}
The map $F$ is strongly monotone and Lipschitz continuous as stated in the following lemma:
\begin{lemma}\label{monolip}
    The pseudo-gradient mapping $F$ in \eqref{eq_pgm} is
    \begin{enumerate}[label=\roman*)]
        \item  strongly monotone if $\alpha<\frac{2}{\kappa(N-1)}$, namely, for any $\bar{\beta},\tilde{\beta}\in K$, 
        ${(\bar{\beta}-\tilde{\beta})^\top(F(\bar{\beta})-F(\tilde{\beta}))\geq \eta_F\|\bar{\beta}-\tilde{\beta}\|^2}$, where 
        \begin{equation}
            \eta_F=\frac{1}{\alpha N}-\frac{\kappa(N-1)}{2N},
        \end{equation}
        \item  Lipschitz continuous with constant
        \begin{equation}
        \kappa_F=\frac{N-1}{N}(\kappa+\frac{1}{\alpha}).
        \end{equation}
    \end{enumerate}
\end{lemma}
The proof of Lemma~\ref{monolip} is provided in Appendix.
Now, we state the main result of this subsection.
\begin{proposition}\label{vGNE}
Assume that $\alpha<\frac{2}{\kappa(N-1)}$.Then the game $\mathcal{G}$ has a unique v-GNE, which is given by the unique solution $\beta^* \in K$ to the variational inequality 
\begin{equation}\label{eq_vi}
(\beta-\beta^*)^\top F(\beta^*)\geq 0, \quad \forall \, \beta\in K.
\end{equation}
\end{proposition}
\begin{proof}
The mapping $F$ is strongly monotone by
Lemma~\ref{monolip} item~{i}. Note that the set $K$ is convex and closed. It follows from \cite[Theorem~2.3.3]{facchinei2003finite} that VI$\big(K,F\big)$ has a unique solution, corresponding to the unique v-GNE of the game $\mathcal{G}$ \cite{kulkarni2012variational}.
\end{proof}

\begin{remark}
{One can show that the condition assumed in Proposition \ref{vGNE} is not required for proving \textit{uniqueness} of the v-GNE. In other words, the v-GNE, if exists, is unique. We resort to this assumption since, in any case, the convergence results of Section \ref{algorithm} hinges on the strong monotonicity of the map $F(\cdot)$ established in Lemma \ref{monolip}.} 
\end{remark}

\begin{remark} \label{rm:extension}
We note that the above market model and the subsequent analyses can be extended to the case where the market is cleared at multiple time slots providing that there is no coupling between the time slots or if any coupling constraint, e.g. those corresponding to energy storage devices, appears linearly \cite{belgioioso2022operationally}.
\end{remark}

\vspace{-0.6cm}
\subsection{Efficiency analysis}
The self-interested behaviors of energy consumers can  lead to market inefficiency. To analyze this efficiency loss in  the game $\mathcal{G}$, we first introduce a social welfare maximization problem as the benchmark. From a social  point of view, it is desirable to utilize the available flexible resources in such a way that the total cost/disutility of active energy consumers is minimized. If consumers are willing to cooperate and reveal their actual economic and technical characteristics to the utility company and the DSO, then the  efficient flexibility allocation can be found as the solution to the following social welfare optimization problem:
\begin{equation}\label{eq_optimization1}
\min_{x} \sum_{n\in \mathcal{N}} C_n(x_n)\ \ \mathrm{s.t.}\quad \eqref{eq_flexibility_balance}, \eqref{eq_pf_compact}, \eqref{eq_bus_2} \ \mathrm{and} \ \eqref{eq_constraint_original}. 
\end{equation}

The following lemma relates the aforementioned social welfare optimization problem with the allocated feasibility at the v-GNE of the game $\mathcal{G}$.
\smallskip{}
\begin{lemma}\label{equivalence}
Let $\beta^*$ be the v-GNE of the game $\mathcal{G}$, i.e, the unique solution to the VI in \eqref{eq_vi},  
and $x^*$ be the corresponding allocated flexibility at this point, namely $x^*=A\beta^*+b$. Then, $x^*$ is the unique solution of the following optimization problem:
\begin{equation}\label{eq_optimization2}
\min_{x} \sum_{n\in \mathcal{N}} D_n(x_n)\ \ \mathrm{s.t.}\quad \eqref{eq_flexibility_balance}, \eqref{eq_pf_compact}, \eqref{eq_bus_2} \ \mathrm{and} \ \eqref{eq_constraint_original}.
\end{equation}
where $D_n(x_n)=C_n(x_n)+\frac{x^2_n}{2\alpha(N-1)}$.
\end{lemma}
The proof of Lemma \ref{equivalence} is provided in Appendix.

\begin{remark} The above lemma indicates that the allocated flexibility at the v-GNE of the game $\mathcal{G}$, i.e.  $x^*$, is different from the solution to the social welfare optimization  \eqref{eq_optimization1}. The additional nonnegative term $\frac{x^2_n}{2\alpha(N-1)}$ is due to the strategic bidding behavior of self-interested consumer $n$ and $D_n(x_n)$ can be regarded as a fake cost function that consumer $n$ submits to the utility company to gain more profit.  
\end{remark}

We use Price of Anarchy ($\mathrm{PoA}$) \cite{paccagnan2022utility} to measure the efficiency loss of the game $\mathcal{G}$, which indicates  how the overall efficiency of a game degrades due to the strategic behavior of consumers. $\mathrm{PoA}$ is defined as the ratio of the total cost between the Nash Equilibrium and the social optimum. In our case, this gives rise to
\begin{equation}\label{eq_poa1}
{\mathrm{PoA}}=\frac{\sum_{n\in\mathcal{N}} C_n(x^*_n)}{\sum_{n\in\mathcal{N}} C_n(\bar{x}_n)},   
\end{equation}
where $\bar{x}$ is the solution to the social welfare optimization problem \eqref{eq_optimization1}, 
and $x^*$ is the allocated flexibility at v-GNE of the game $\mathcal{G}$.
We then have the following result.
\smallskip{}
\begin{proposition}\label{poa}
The price of anarchy in \eqref{eq_poa1} satisfies
\begin{equation}\label{eq_poa}
    \mathrm{PoA}<1+\frac{1}{2\alpha (N-1)}\frac{\sum_{n\in\mathcal{N}} \bar{x}^2_n}{\sum_{n\in\mathcal{N}} C_n(\bar{x}_n)}\,.
\end{equation}
\end{proposition}
\begin{proof}
We observe that
\begin{equation}
\sum_{n\in\mathcal{N}} C_n(x^*_n)<
\sum_{n\in\mathcal{N}} D_n(x^*_n)\leq \sum_{n\in\mathcal{N}} D_n(\bar{x}_n),    
\end{equation}
where the first inequality follows from the positivity of $\sum_{n\in\mathcal{N}}\frac{(x^*_n)^2}{2\alpha(N-1)}, $\footnote{Note that due to \eqref{eq_constraint_original} and \eqref{eq_flexibility_balance}, there exists $n\in\mathcal{N}$ with $x^*_n>0$.} and the second inequality follows from the fact that $x^*$ is the solution to the optimization problem \eqref{eq_optimization2}.
Hence, by using the definition of $D_n(\cdot)$, we obtain that 
\begin{equation}
  \sum_{n\in\mathcal{N}} C_n(x^*_n)<\sum_{n\in\mathcal{N}} \big(C_n(\bar{x}_n)+\frac{\bar{x}^2_n}{2\alpha(N-1)}\big),  
\end{equation}
which leads us to \eqref{eq_poa}.
\end{proof}

\begin{remark}
The result of the above proposition provides an upper bound for $\mathrm{PoA}$. This bound largely depends on the number of  active energy consumers $N$ and the parameter $\alpha$. For sufficiently large quantity $\alpha (N-1)$, the upper bound gets close to $1$ and the market power of each individual consumer decreases.
%and the \mathrm{PoA} approaches to $1$ . 
\end{remark}

\section{Algorithm design}\label{algorithm}
In this section, we present a decentralized market clearing method to compute the v-GNE of the game $\mathcal{G}$ in \eqref{eq_game}. Both the utility company and the DSO are willing to engage in this algorithm, because at v-GNE, energy consumers are unlikely to deviate unilaterally from their promised bids and the physical network constraints are maintained. We first explain the restrictions in sharing information and knowledge available to each party in our energy balancing problem. As for the consumers, they are not willing to reveal their economic and technical specifications, including the cost function $C_n(x_n)$ and the maximum available flexibility $\hat{x}_n$, to the utility company, the DSO or any other consumers. The utility company is not willing to share the load adjustment requirement $\xf$ with the consumers. Finally, the physical specifications of the system, such as network topology and line parameters are only known to the DSO. 

Motivated by the above information sharing limitations, we split the feasible set $K$ in \eqref{eq_game}, as $K=\Lambda \cap \Psi$, where $\Lambda=\{\beta\in \mathbb{R}^{N} \mid \eqref{eq_global_set1}\,  \text{holds} \}$ and $\Psi=\{\beta\in \mathbb{R}^{N} \mid \eqref{eq_pf_compact},\, \eqref{eq_global_set2}\, \text{and}\,\eqref{eq_global_set3} \, \text{hold}\}$. Note that the former contains the private information of consumers $\{\hat x_n\}_{n\in \mathcal{N}}$ and will be handled by the consumers themselves in the algorithm. However, the latter set involves the non-private information of consumers alongside the DSO physical network constraints and will be handled by the DSO. 

\vspace{.3cm}
\begin{algorithm}
\caption{Decentralized market clearing} \label{alg1}
\textbf{Initialization:} $\forall \, n\in\mathcal{N}$, set $\beta_n^0$ and $\gamma_n^0\in \mathbb{R}^+$ and choose step sizes $\rho_n,\nu_n\in \mathbb{R}^+$.

\textbf{Iterate until convergence:} 
\begin{enumerate}[leftmargin=*]
    \item $\forall\, n\in\mathcal{N}$ (In parallel): 

 \qquad Update bid $\tilde{\beta}_n^{k+1}$s using \eqref{eq_update_beta},
 
 \qquad Communicate $\tilde{\beta}_n^{k+1}$s to the DSO. 
    \item The DSO:

  \qquad Modifies bids  $\tilde{\beta}_n^{k+1}$s using \eqref{dso},

 \qquad Send  $\beta_n^{k+1}$s to the utility company.
     \item The utility company:
     
     \qquad Updates price $\lambda^{k+1}$ using  \eqref{eq_price}.

      \qquad Broadcasts price $\lambda^{k+1}$ to consumers.
    \item $\forall\, n\in\mathcal{N}$ (In parallel): 
    
    \qquad Update dual variable $\gamma_n^{k+1}$s using  \eqref{eq_update_gamma},
      
    \qquad Send $\gamma_n^{k+1}$s to the utility company.  
    \end{enumerate}
\end{algorithm}
\vspace{.3cm}
The market clearing process is summarized in Algorithm \ref{alg1}. In the iterative bidding process, each energy consumer $n$ is required to update two variables $\tilde{\beta}_n$ and $\gamma_n$. The variable $\tilde{\beta}_n$ is an intermediate variable and can be interpreted as the intended and unverified bid of energy consumer $n$. The update rule for the variable $\tilde{\beta}_n$ is defined as follows, 
\begin{equation}\label{eq_update_beta}
              \tilde{\beta}_n^{k+1}  = \beta_n^k - \rho_n h_n^k,
          \end{equation}
where $k$ indexes the time step and $h_n^k$ is calculated as
\begin{equation}\label{h}
    \begin{aligned}
     h_n^k&=C'_n(\alpha \lambda^k+\beta_n^k)\frac{N-1}{N} + \frac{\alpha \lambda^k (2-N)+\beta_n^k}{\alpha N}\\
     & -\frac{1}{N} \mathbbb{1}^{\top} \gamma^k + \gamma_n^k,
     \end{aligned}           
\end{equation}
with $\gamma=\col(\gamma_n)_{n\in\mathcal{N}}$. The dual variable $\gamma_n$ is associated with the Lagrangian multiplier of the coupling constraint \eqref{eq_global_set1}. %The introduction of this dual variable is to avoid that the resulting flexibility allocation $x_n$ deviates the constraint \eqref{eq_global_set1}. 
The update rule for this variable is  
\begin{multline}\label{eq_update_gamma}
     \gamma_n^{k+1}=\proj_{\mathbb{R}^+}(\gamma_n^k + \nu_n(2x_n^{k+1}-x_n^k-\hat{x}_n)).
\end{multline}
Note that $x_n$ can be determined locally using the expression $x_n=\alpha \lambda +\beta_n$, provided that the price is communicated by the utility company.

Next, the DSO performs the validation and necessary corrections of all the intended bids (i.e., $\tilde{\beta}=\col(\tilde{\beta}_n)_{n\in\mathcal{N}}$) by solving an optimization problem:
\begin{equation}\label{dso}
            \beta^{k+1} = \argmin_{z\in \Psi} ||z-\tilde{\beta}^{k+1}||.
          \end{equation}
The solution of this problem ensures compliance with the constraints specified in set $\Psi$. Note that  $\beta=\tilde{\beta}$ if the intended bids are feasible for DSO; otherwise the bids are modified to the ``closest" physically feasible ones. The utility company is responsible for updating the price.  
%Next, we demonstrate that only non-private information is exchanged among the participants as outlined in Algorithm \ref{alg1}.

 The communication at each iteration is described as follows:
\begin{itemize}
    \item For bid verification and modification, the DSO collects $\tilde{\beta}^{k+1}_n$s from energy consumers or the utility company collects them and relays them to the DSO . 
    \item The DSO dispatches the authenticated bids $\beta^{k+1}_n$s to the utility company for updating the price.
    \item The utility company sends $\lambda^{k+1}_n$s to each energy consumer $n\in\mathcal{N}$ to update the dual variables $\gamma_n$s.
    \item The utility company aggregates $\gamma^k_n$s from consumers and transmits their cumulative sum $\mathbbb{1}^{\top} \gamma^{k}$ to each energy consumer $n\in\mathcal{N}$ to facilitate the updating of the bids.
\end{itemize}

While global information such as 
$\alpha$ and the total number of active consumers $N$
is necessary for the implementation of the algorithm, these parameters are typically public and exhibit infrequent alterations. The parameters $\rho_n$ and $\nu_n$ are step sizes which should be sufficiently small to ensure convergence. The following result provides upper bounds for the step sizes $\rho_n$ and $\nu_n$ in Algorithm~\ref{alg1} such that convergence to v-GNE of the game is guaranteed. 
\begin{proposition}\label{convergence}
Consider the monotonicity and Lipschitz continuity constants $\eta_F$ and $\kappa_F$ given in Lemma~\ref{monolip}.
If
\begin{equation}\label{eq_step_sizes}
\frac{\kappa_F^2}{2\eta_F}< \frac{1}{\bar \rho} - \bar \nu,     
\end{equation}
with $\bar\rho=\max\{\rho_n\}_{n\in\mathcal{N}}$
and $\bar\nu=\max\{\nu_n\}_{n\in\mathcal{N}}$,
then the vector $\omega^k=\col(\beta^k, \gamma^k)$ in Algorithm~\ref{alg1}  sublinearly converges to $\omega^*=\col(\beta^*, \gamma^*)$, where $\beta^*$ is the v-GNE of the game, and
\begin{equation}\label{eq_rate}
    \min_{j=0,1,...,k} \|\omega^{j+1}-\omega^j\|^2  \leq \mathcal{O}(1/k).
\end{equation}
\end{proposition} 
The condition in \eqref{eq_rate} illustrates Algorithm \ref{alg1} generates sequences for which $\|\omega^{k+1}-\omega^k\|^2$ converges to zero arbitrarily closely at a rate of $ \mathcal{O}(1/k)$. This result is consistent with the findings in \cite[Chapter 5.2]{boyd2016primer} for algorithms utilizing averaged operators. We prove the convergence of this algorithm by showing it is essentially a preconditioned forward-backward iteration. The details are provided in Appendix.
\section{Case study}
We perform an extensive numerical study on the modified IEEE 33-bus distribution network \cite{dolatabadi2020enhanced} shown in Fig.~\ref{fig:simulation_network} to validate the proposed market. The pre-scheduled energy generation and consumption profiles are also based on \cite{dolatabadi2020enhanced}. For practical reasons, this benchmark network is enhanced in several aspects. In particular,  the ratio of lines' reactance to resistance is decreased, reactive power compensators are integrated and more strict voltage limits are imposed. 
\begin{figure}[ht]
\begin{center}
\includegraphics[width=0.48\textwidth]{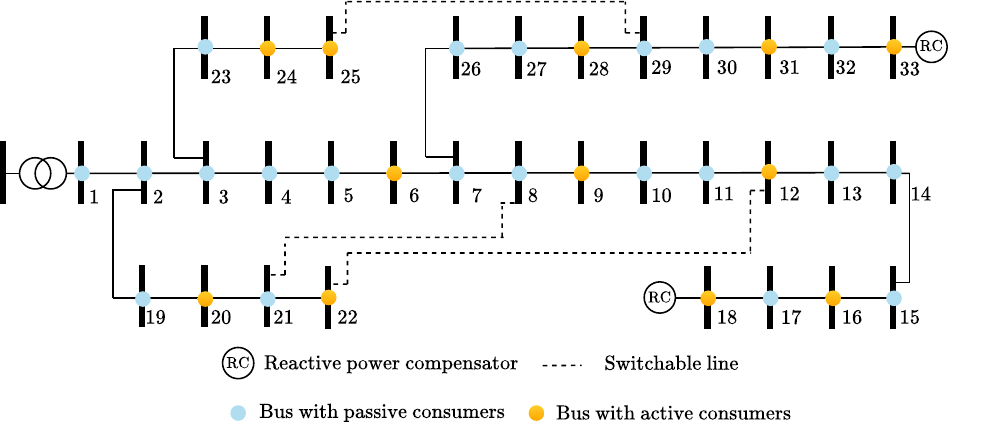}
\end{center}
\caption{IEEE 33-bus distribution network}
\label{fig:simulation_network}
\end{figure}

We consider linear-quadratic cost functions for consumers as $C_n(x_n)=\frac{1}{2}a_nx_n^2+b_nx_n$. The parameters $a_n$ and $b_n$ are arbitrarily selected in the intervals $[0.003, 0.005]\$/(\mathrm{kWh})^2$ and $[0.35, 0.45]\$ /\mathrm{kWh}$, respectively. 
These cost functions satisfy Assumption~\ref{assumption} with $\kappa=0.005$.
The maximum flexibility capacities, $\{\hat{x}_n\}_{n\in \mathcal{N}}$ are also arbitrarily selected with $\sum_{n\in\mathcal{N}} \hat{x}_n\geq \xf=100\mathrm{kWh}$. The utility company can choose the parameter in the interval $\alpha=\delta \frac{2}{\kappa(N-1)}$ with $\delta \in (0,1)$. Note that the latter choice satisfies the condition in Proposition \ref{vGNE}. To ensure the convergence condition \eqref{eq_step_sizes} holds, the step sizes $\rho_n$ and $\nu_n$ are chosen such that $\bar{\rho}=c\frac{2\eta_F}{\kappa_F^2}$ and $\bar{\nu}=0.8\left(\frac{1}{c}-1\right) \frac{2\eta_F}{\kappa_F^2}$ with $c\in (0,1)$. In the following, we first show the economic efficiency at v-GNE and briefly examine the effect of a consumer's deceptive behavior. Next, we evaluate the importance of having physical network constraints in our model. Finally, we test the effectiveness of our proposed algorithm.

\subsection{Efficiency analysis} \label{subsec:EA}
In this part, we mainly focus on the advantages of our game in terms of the market efficiency performance. We do not incorporate the physical network constraints in our model here, to make it consistent with the setup in \cite{li2015demand} and \cite{xu2015demand} where the physical network is neglected. 
The three cases are as follows,

\noindent Case 1 (C1): Each consumer reveals its true cost function and constraints to the utility company. In this case, we solve the social welfare maximization problem \eqref{eq_optimization1}.

\noindent Case 2 (C2): Each consumer chooses the supply function in \cite{li2015demand} and \cite{xu2015demand} as its bidding strategy.

\noindent Case 3 (C3): Each consumer chooses the y-intercept of its linear supply function  ($\beta_n$) as its bidding strategy based on our proposed market model.  

Moreover, we study two scenarios for each case. In Scenario~1, we consider only the nonnegative lower bound on the available flexibility $x_n$ to be consistent with \cite{li2015demand}, whereas in Scenario 2 we consider also the upper bound $\hat{x}_n$ as in \cite{xu2015demand}. Beyond the PoA, our analysis incorporates two additional indices for assessing market efficiency losses: the Lerner Index (LI) \cite{johari2011parameterized} and Deadweight Loss (DWL) \cite{daskin1991deadweight}. This approach allows for a thorough evaluation of market efficiency across various scenarios. The LI is commonly used to measure the price markup above competitive levels in oligopolies, from the individual perspective. The DWL evaluates the difference between market outcome and the social optimum from the whole perspective. We generate ten sets of consumers' parameters randomly for each scenario and subsequently calculate the average values of the market efficiency loss indices.   
\begin{figure}[ht]
\begin{center}
\includegraphics[width=0.45\textwidth]{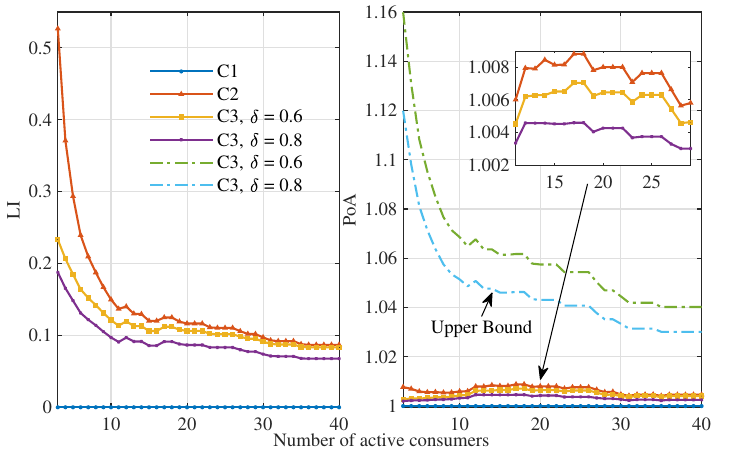}
\end{center}
\caption{The LI and PoA in Scenario 1}
\label{fig:simulation1}
\end{figure}

\begin{figure}[ht]
\begin{center}
\includegraphics[width=0.45\textwidth]{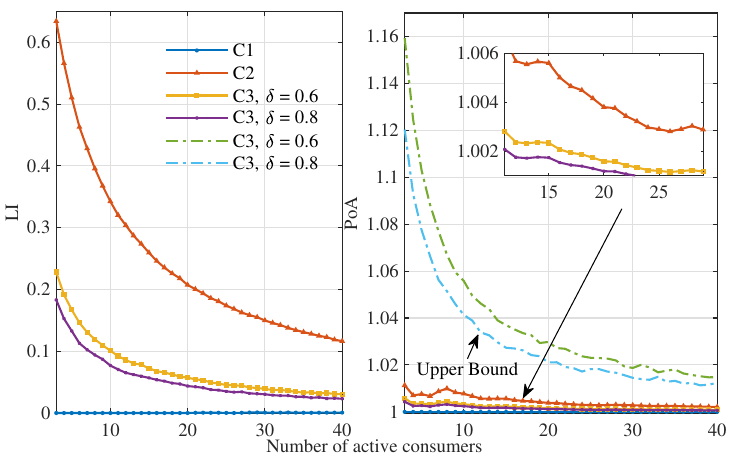}
\end{center}
\caption{The LI and PoA in Scenario 2}
\label{fig:simulation12}
\end{figure}

\begin{figure}[ht]
\begin{center}
\includegraphics[width=0.45\textwidth]{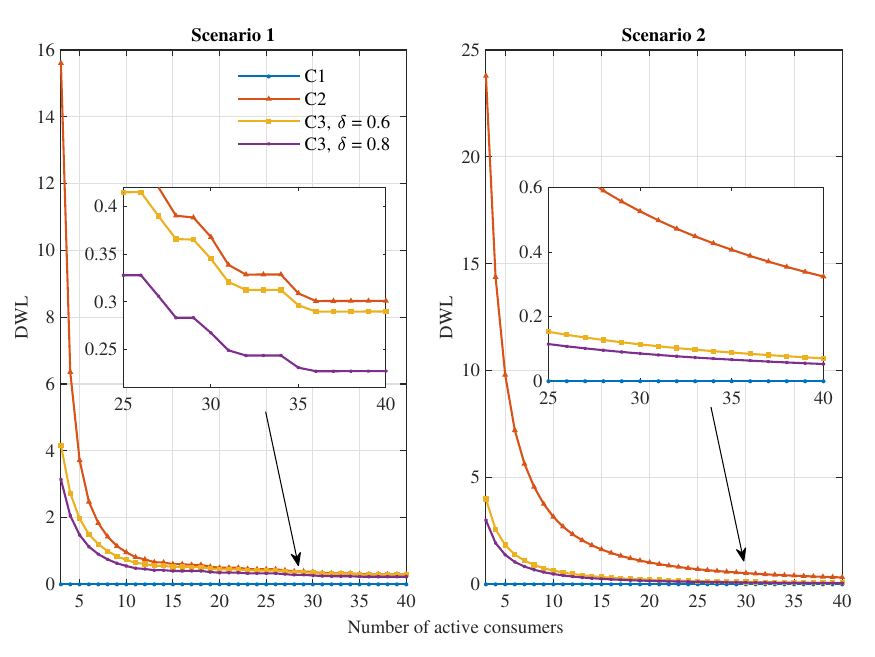}
\end{center}
\caption{The DWL in Scenario 1 and 2}
\label{fig:simulation13}
\end{figure}

Fig. \ref{fig:simulation1} and Fig. \ref{fig:simulation12} show the LI and PoA in Scenario 1 and Scenario 2, respectively for various numbers of active consumers. From the figures, it is evident that for Case 1, which is a perfectly competitive market, the LI remains constant at $0$ and the PoA at $1$, as expected, irrespective of the number of active consumers. 
In the subsequent three cases, the LI decreases as the number of active consumers rises. This trend is attributable to the diminishing market power of individual participants in an imperfectly competitive market, which naturally occurs as the number of market participants increases. 
Additionally, our analysis reveals that in both scenarios, the LI and PoA in Case 3, which showcases our proposed mechanism, approach the ideal benchmarks of zero and one, respectively.
In particular, in Scenario 2, where upper bounds on local flexibility constraints are considered, the LI in Case 2  is significantly higher than in Case 3. Specifically, in Scenario 1, the average LI and PoA in Case 2 are 22.15\% and 0.28\% higher than those in Case 3 ($\delta=0.6$), whereas in Scenario 2, these figures are 232\% and 0.31\%, respectively. This difference underscores the superior performance of our proposed market mechanism.
Bearing in mind that $\alpha$ scales linearly with respect to the parameter $\delta$, we observe that the proposed mechanism performs better in terms of market efficiency, as $\alpha$ becomes larger. Additionally, it is important to note that the actual PoA is substantially lower than its theoretical upper bound, as derived in Proposition \ref{poa}. 
These upper bounds are visually represented by the dashed lines in Fig. \ref{fig:simulation1} and \ref{fig:simulation12}.
Ultimately, we demonstrate the DWL in Figure \ref{fig:simulation13}. Clearly, the DWL shows patterns similar to the LI.

\subsection{Deceptive behaviors effect}\label{subsec:deceptive}
Here, we briefly examine what happens if a consumer deviates from 
Algorithm~\ref{alg1}. In particular, we model such a deceptive behavior by modifying the true cost function parameters $a_n$ or $b_n$ to $\hat a_n$ or $\hat b_n$, respectively.
We consider a group of 5 consumers where Consumer 1 shows deceptive behavior while the other four consumers remain truthful. Fig.~\ref{fig:simulation15} shows the profit of Consumer 1 when using deceptive cost parameters $\hat a_1$ or $\hat b_1$, across three distinct cases. Note that in all cases, the true cost parameters of Consumer 1 remain the same, whereas the true cost parameters for the remaining four consumers are chosen differently for each case. %As an illustration, the cost parameters of Consumer 2 in case 1 differ from those in case 2. 
We observe from Fig.~\ref{fig:simulation15} that in Case 1, Consumer 1 achieves higher profits by using slightly larger $\hat{a}_1$ or $\hat{b}_1$. On the contrary, in Case 2 a slightly smaller $\hat a_1$ and in Case 3 the true value $\hat a_1=a_1$ is preferred. In both Case 2 and Case 3, adopting the true value $\hat b_1=b_1$ is more advantageous.   %the other two cases, it is generally more advantageous to adopt approximately the true $a_1$ and  $b_1$.
%with the exception that using a slightly smaller $\hat{a}_1$ in Case 2 proves to be more profitable. 
%We observe from Fig.~\ref{fig:simulation15} that in Case 1, Consumer 1 will achieve higher profits by using a bit larger $\hat{a}_1$ or $\hat{b}_1$, whereas in other two cases, it is more advantageous to adopt the true $a_1$ or $b_1$, except that it more profitable to use a bit smaller $\hat{a}_1$ in Case 2. 
Since  these cases are different in terms of the cost parameters of other consumers, we conclude that whether Consumer 1 will gain an extra profit or incur a profit loss as a result of using a fake $\hat{a}_1$ or $\hat{b}_1$ depends on the cost parameters of the other consumers. Due to privacy considerations, it is very difficult for the deceptive consumer to obtain the cost parameters of other consumers, and thus it cannot determine whether a larger or smaller $\hat{a}_1$ or $\hat{b}_1$ will guarantee a profit gain.  This weakens the motivation of the consumers to deviate from the nominal protocol.

\begin{figure}[ht]
\begin{center}
\includegraphics[width=0.45\textwidth]{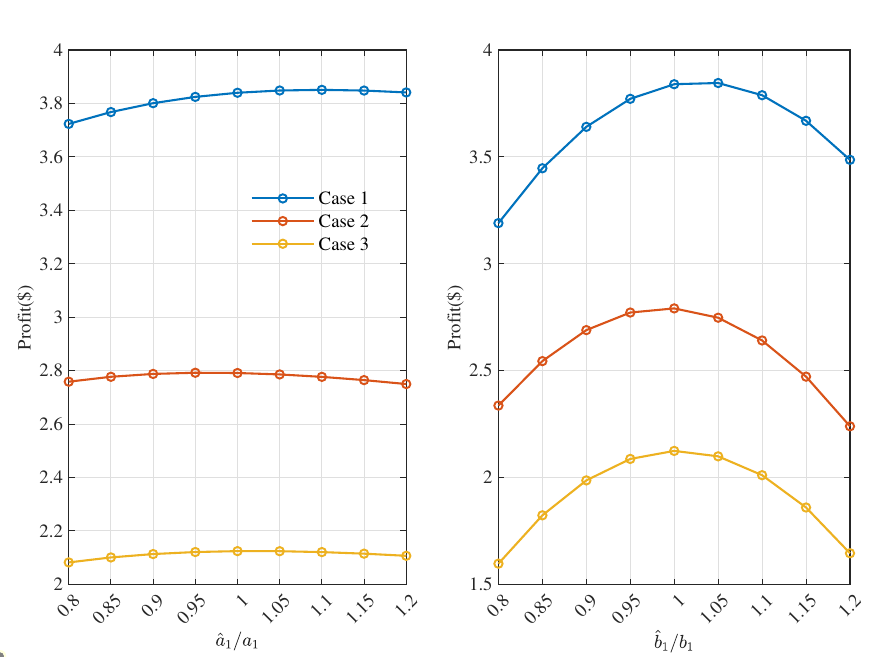}
\end{center}
\caption{The profit of Consumer 1 with deceptive value $\hat a_1/a_1$ (left) and $\hat b_1/ b_1$ (right)}
\label{fig:simulation15}
\end{figure}

\subsection{Security analysis} \label{subsec:SA}
Next, we highlight the importance of including AC power flow network constraints by comparing the results with and without including such constraints in our model. We consider 12 consumers providing balancing services for the utility company. For the sake of simplicity, we assume only one consumer is connected to each bus and we label each active consumer by the bus number it is connected to. We investigate two scenarios:

\noindent Scenario 1: There is an energy surplus. In this scenario, to consume the extra energy from the transmission system, an active consumer with flexible loads is rewarded to consume more energy.

\noindent Scenario 2: There is an energy deficit. In this scenario, an active consumer with distributed generators is getting paid to provide extra energy for the utility company.

Fig. \ref{fig:simulation2} shows the flexibility allocation and the corresponding voltage magnitudes in Scenario 1. We can see if the restriction on the voltage magnitudes is not considered, the voltage magnitude of Bus 30 drops below the (standard) lower bound $\underline{v}=0.95$. On the other hand, when network constraints are incorporated in the design,  Consumer 28, 31 and 33 provide less flexibility while Consumer 20, 22, 24 and 25 provide more, thereby avoiding the voltage collapse at Bus 30. The flexibility allocations and the corresponding power flows in Scenario 2 are shown in Fig. \ref{fig:simulation22}. These results illustrate that to avoid Line 17 (connecting buses 17 and 18) being congested, Consumer 18 should decrease its flexibility significantly from 7.41 $\mathrm{kWh}$ to 2.28 $\mathrm{kWh}$. In practice, this means active Consumer 18 cannot provide more flexibility since it has provided a large amount of energy in the day-ahead market. In general, incorporating the physical constraints might extensively alter the flexibility allocation results, and clearing the market without considering the physical constraints can jeopardize the system's secure operation.

\begin{figure}[ht]
\begin{center}
\includegraphics[width=0.48\textwidth]{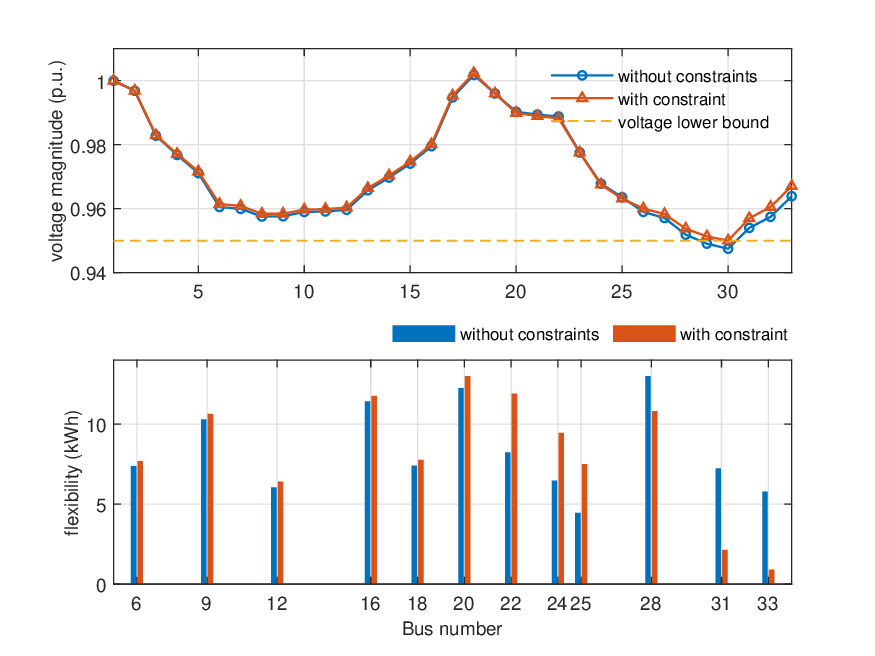}
\end{center}
\caption{The voltage magnitudes and allocated flexibility in Scenario 1}
\label{fig:simulation2}
\end{figure}

\begin{figure}[ht]
\begin{center}
\includegraphics[width=0.48\textwidth]{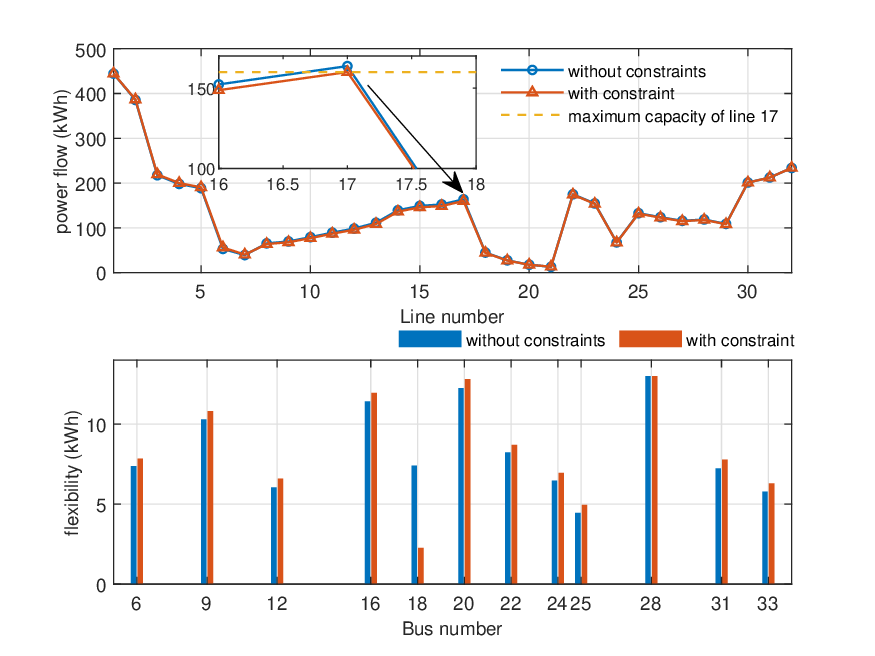}
\end{center}
\caption{The power flow and allocated flexibility in Scenario~2}
\label{fig:simulation22}
\end{figure}

\begin{figure}[ht]
\begin{center}
\includegraphics[width=0.48\textwidth]{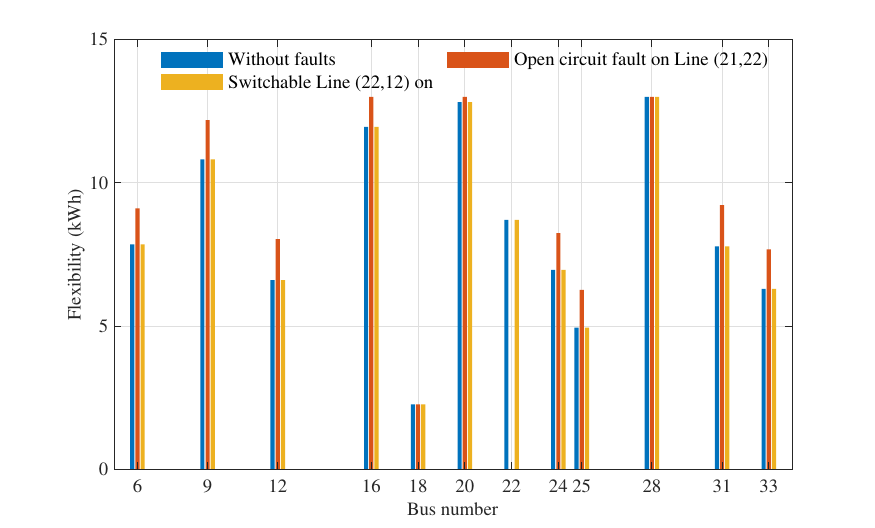}
\end{center}
\caption{The allocated flexibility with fault awareness}
\label{fig:simulation23}
\end{figure} 

Due to a fault occurring in the distribution system, the topology of the network could be changed resulting in different market clearing results. To further evaluate our proposed market mechanism, we consider a situation based on Scenario 2 above where an open circuit fault occurs on Line (21,22), and switchable Line (25,29) (embedded in the enhanced IEEE 33-bus system \cite{dolatabadi2020enhanced}) is connected as a response to this fault.
Fig.~\ref{fig:simulation23} shows the market clearing results under different situations. When there is no fault, the flexibility allocation is the same as the one in Fig.~\ref{fig:simulation22}. However, when Line (21,22) is open, Consumer 22 is not able to provide flexibility while other consumers need to provide more flexibility. When Line (25,29) is switched on as the backup line, it reconnects Consumer 22 to the network, allowing Consumer 22 to provide flexibility again. In practice, a DSO is aware of the faults occurring in the network and can change the power flow network constraints in the bid modification phase. Therefore, the market clearing results can satisfy the modified network constraints.

\subsection{Convergence properties of Algorithm \ref{alg1}}\label{subsec:convergence}
We use the same setup as in Scenario 2 of the previous subsection to show the transient performance of the algorithm.
To facilitate comparison, we assume all consumers adopt the same step sizes as $\rho_n=\bar{\rho}$ and $\nu_n=\bar{\nu}$. Hence, a larger $c$ means a larger step size for bid updates and a smaller step size for dual variable updates.
Fig.~\ref{fig:betagamma33large} and Fig.~\ref{fig:betagamma33small} show the evolution of the bids $\beta$ and the dual variables $\gamma$ under different step sizes, where the dual variables $\gamma$ are associated with the flexibility capacity constraints.  As illustrated in the figures, the algorithm converges within $400$ iterations for $c=0.4$ and requires only $150$ iterations for $c=0.8$. In addition, the dual variable corresponding to Consumer 28 takes positive values as soon as its allocated flexibility reaches its maximum available flexibility. Furthermore, as demonstrated in Fig. \ref{fig:simulation22}, line capacity constraints limit Consumer 18's strategy to $-6.21$, resulting in a low flexibility allocation for this consumer.

To demonstrate that this algorithm consistently converges to the v-GNE rather than any other fixed points, Fig. \ref{fig:simulation32} presents the evolution of the normalized errors of the flexibility profile, defined as $\|x^k-x^*\|^2/\|x^*\|^2$. 
Here, $x^*$ represents the flexibility allocation at the v-GNE for the game $\mathcal{G}$, which can be determined by solving the optimization problem \eqref{eq_optimization2}.
From the figure, it is evident that all error signals converge to zero following a brief transient phase, with a larger value of 
$c$ resulting in faster convergence. Therefore, selecting larger $c$ (larger step sizes for bid updates) can accelerate the convergence. 

\begin{figure}[ht]
\begin{center}
\includegraphics[width=0.48\textwidth]{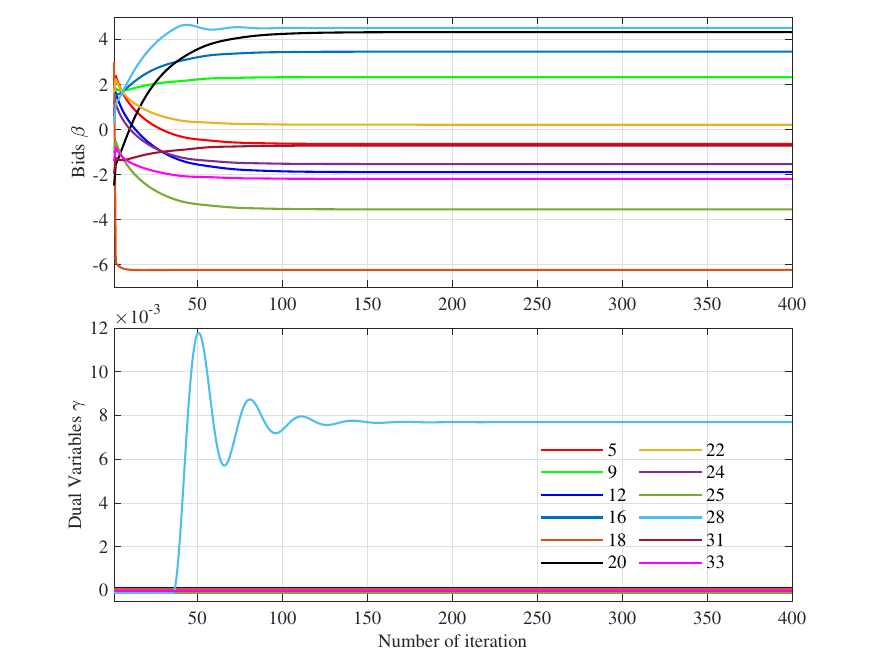}
\end{center}
\caption{The evolution of $\beta$ and $\gamma$ when $c=0.8$}
\label{fig:betagamma33large}
\end{figure}

\begin{figure}[ht]
\begin{center}
\includegraphics[width=0.48\textwidth]{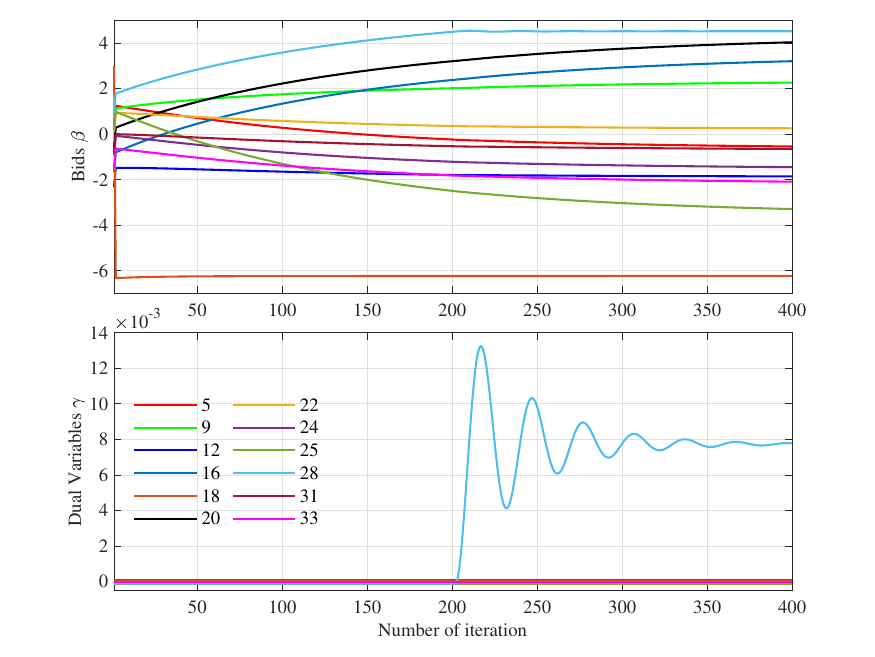}
\end{center}
\caption{The evolution of $\beta$ and $\gamma$ when $c=0.4$}
\label{fig:betagamma33small}
\end{figure}
 %\vspace{.2cm}
\begin{figure}[ht]
\begin{center}
\includegraphics[width=0.48\textwidth]{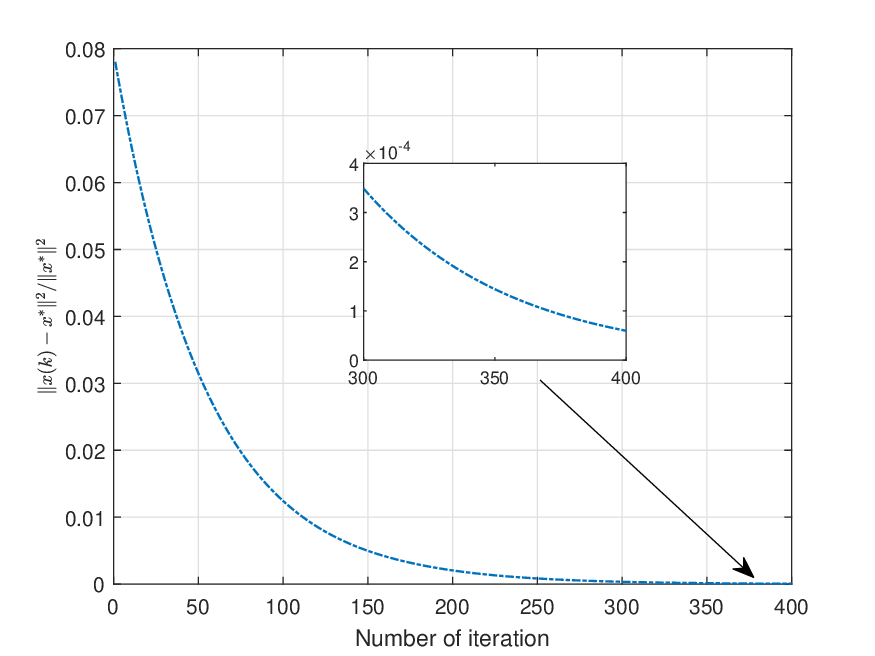}
\end{center}
\caption{The normalized errors of flexibility profile}
\label{fig:simulation32}
\end{figure}

Finally, we investigate the scalability of our algorithm by showing the required computational time for convergence. To this end, we implement our market clearing method for various numbers of active consumers on the IEEE 33-bus system and IEEE 69-bus system \cite{zimmerman2010matpower}. We consider a stopping criterion as 
\begin{equation}
\|\beta^{k+1}-\beta^k\|^2+\|\gamma^{k+1}-\gamma^k\|^2<10^{-5},
\end{equation} 
and assign active consumers to the buses of distribution arbitrarily. Furthermore, to show the impact of the computational complexity of the AC power flow model on the scalability of the proposed method, we additionally employ the SOCRD model to represent the distribution network \cite{farivar2013branch}.  Figure \ref{fig:computationtime} demonstrates that the computational time required scales approximately linearly with the number of active consumers in each network and under each power flow model. For both networks, the computation time fits well within the clearance timescale of the market (typically around 5 minutes). In addition, the average computational time on the IEEE 69-bus system is 87.77\% higher than that on IEEE 33-bus system. It is also worth mentioning that our method targets imperfectly competitive markets where the number of participants is limited. Furthermore, employing the SOCRD power flow model requires averagely 2.67 times the computational time compared to the linear power flow model \eqref{eq_line_balance}. 
\begin{figure}[ht]
\begin{center}
\includegraphics[width=0.48\textwidth]{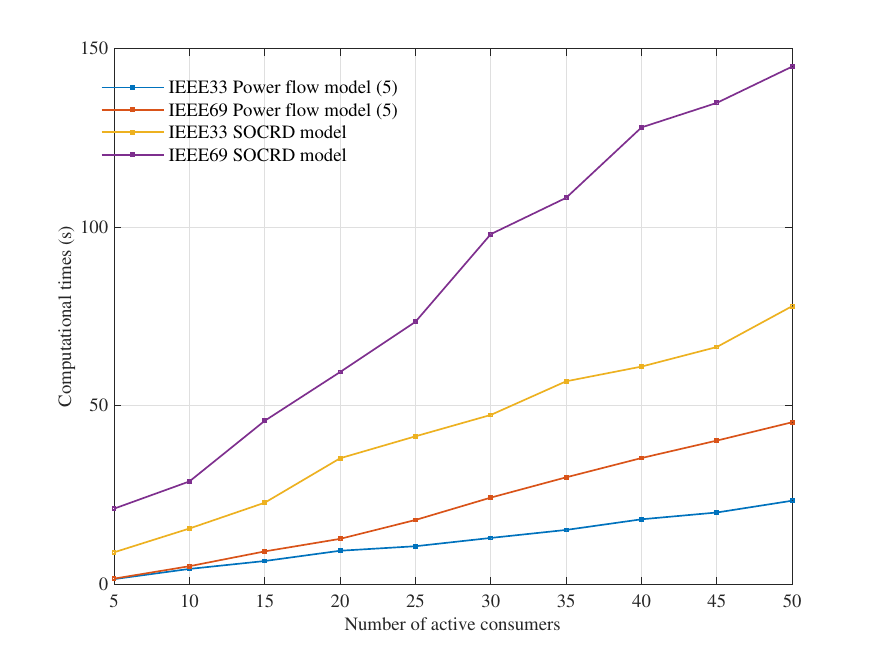}
\end{center}
\caption{The computation time}
\label{fig:computationtime}
\end{figure}

\section{Conclusion}
We have presented a demand response game for procuring energy balancing services within the distribution network and proposed a scalable decentralized algorithm to solve the formulated game. The numerical experiments have shown that this algorithm leads the system to the v-GNE, where the market efficiency loss is small and the physical network constraints are satisfied. Future research avenues include incorporating uncertainties into the model, such as the uncertain maximum flexibility of consumers, as well as devising mechanisms for achieving social welfare maximization through efficient profit sharing. Furthermore, there is a need to address the challenges posed by non-continuous cost functions.

\appendix
\begin{proof}[Proof of Lemma~\ref{monolip}]
The vector of allocated flexibility $x=\col(x_n)_{n\in\mathcal{N}}$,  derived from \eqref{eq_flexibility_allocation}, can  be compactly written as
\begin{equation}\label{eq_flexibility_allocation_compact}
 x=A\beta+b,
\end{equation}
where $A= I- \frac{1}{N} \mathbbb{1}\mathbbb{1}^{\top}$ and $b=\frac{\xf}{N}\mathbbb{1}$.
We first write the mapping $F$ in a vector form as
    \begin{equation}
        F(\beta)=\frac{N-1}{N}C'(x)+A_F\beta+b_F,
    \end{equation}
where $C'(x)=\col(C_n'(x_n))_{n\in\mathcal{N}}$, $b_F=-\frac{\xf(N-2)}{\alpha N^2}\mathbbb{1}$,  $A_F={\frac{1}{\alpha N^2}\big( NI+(N-2)\mathbbb{1}\mathbbb{1}^{\top} \big)}$ and $x=A\beta+b$ as in \eqref{eq_flexibility_allocation_compact}. Note that $\lambda_{\min} (A_F)=\frac{1}{\alpha N}$, $\lambda_{\max} (A_F)=\frac{N-1}{\alpha N}$ and $\lambda_{\max}(A)=1$.

\noindent    i) The Jacobian matrix of the mapping $F$ is given by 
\begin{equation}
\nabla F(\beta)= \frac{N-1}{N}C''(x)A +A_F,
\end{equation}
where $C''(x)=\diag(C''_n(x_n))_{n\in\mathcal{N}}$ with $C''_n=\nabla C'_n$. Based on \cite{boyd2016primer}, The mapping $F$ is strongly monotone with the constant $\eta_F>0$ if and only if $E(x)+A_F\succeq \eta_F I$, where 
\begin{equation}
 E(x)= \frac{N-1}{N}\frac{C''(x)A+AC''(x)}{2}.   
\end{equation}
Note that elements of the matrix $E(x)$ can be stated as
\begin{equation}
E_{nm}(x)=\frac{N-1}{2N^2}
\begin{cases}
2(N-1)C''_n(x_n) \quad &\mathrm{if} \ n=m, \\
-C''_n(x_n)-C''_m(x_m) \quad &\mathrm{if} \ n\neq m.
\end{cases}    
\end{equation}

By exploiting the structure of $E(x)$ and leveraging the Gershgorin circle theorem \cite{bernstein2009matrix}, we have
\begin{equation}
\lambda_{\min}(E(x)) \ge \frac{N-1}{2N} \big(C''_n(x_n) -\frac{\sum_{m=1}^{N}C''_m(x_m)}{N} \big),
\end{equation}
for all $n \in \mathcal{N}$. Due to the convexity and Lipschtiz continuity in Assumption~\ref{assumption}, if  $\alpha<\frac{2}{\kappa(N-1)}$, we have
\begin{equation}
  \begin{aligned}
    \lambda_{\min} (E(x)+A_F)  & \geq \lambda_{\min} (E(x)) + \lambda_{\min} (A_F) \\
    &  \geq -\frac{N-1}{2N} \kappa+\frac{1}{\alpha N}
\end{aligned}  
\end{equation}
which implies the constant $\eta_F=\frac{1}{\alpha N}-\frac{\kappa(N-1)}{2N}$.

\noindent    ii) Let $\tilde{x}=A\tilde{\beta}+b$ and $\bar{x}=A\bar{\beta}+b$ for any $\tilde{\beta},\bar{\beta} \in K$.
    Then, we can write the following 
\begin{equation}
    \begin{aligned}
        &\lVert F(\tilde{\beta})-F(\bar{\beta})\rVert \leq \lVert (C'(\tilde{x})-C'(\bar{x}))\rVert+\lVert A_F(\tilde{\beta}-\bar{\beta})\rVert \\ \leq 
        & \frac{N-1}{N}\kappa \lVert A(\tilde{\beta}-\bar{\beta}) \rVert + \lVert A_F( \tilde{\beta}-\bar{\beta}) \rVert \\ \leq
        & \frac{N-1}{N}\kappa \lambda_{\max}(A) \lVert \tilde{\beta}-\bar{\beta} \rVert + \lambda_{\max}(A_F)\lVert \tilde{\beta}-\bar{\beta} \rVert\\
        =& \frac{N-1}{N}(\kappa+ \frac{1}{\alpha})\lVert\tilde{\beta}-\bar{\beta} \rVert,    
    \end{aligned}
\end{equation}
for any $\tilde{\beta},\bar{\beta} \in K$. This gives $\kappa_F=\frac{N-1}{N}(\kappa+ \frac{1}{\alpha})$.
\end{proof}
\begin{proof}[Proof of Lemma \ref{equivalence}]
We write the feasible set of the optimization problem \eqref{eq_optimization2} as $\Theta=\{x \mid \eqref{eq_constraint_original},\,\eqref{eq_flexibility_balance} ,\,\eqref{eq_bus_2} \, \mathrm{and} \, \eqref{eq_pf_compact}\, \mathrm{hold}\}$. By \cite[Theorem~3.34]{ruszczynski2011nonlinear}, a feasible point $\bar x$ is an optimizer of \eqref{eq_optimization2} if and only if
\begin{equation} \label{ineq_NC}
    (x-\bar x)^{\top} \nabla_{x} (\sum_{n\in\mathcal{N}} D_n(\bar x_n)) \geq \mathbbb{0}, \ \ \forall \, x \, \in \Theta. 
\end{equation}
Now, we show that $\bar x=A\beta^*+b=x^*$ satisfies the inequality above. By using $\lambda^*=(\xf-\mathbbb{1}^{\top}\beta^*)/(\alpha N)$ and noting \eqref{eq_pgm},  we can write $f_n(\beta^*)$ as
\begin{equation}
    f_n(\beta^*)=\frac{N-1}{N}\left(\nabla_{{x}_n} C_n( x_n^*) + \frac{x^*_n}{\alpha(N-1)}-\lambda^*\right).
\end{equation}
Hence, the VI \eqref{eq_vi} can be written as
\begin{equation}\label{eq_optimality2}
\frac{N-1}{N}(\beta-\beta^*)^\top \Big( \nabla_{x} (\sum_{n\in\mathcal{N}} D_n( x_n^*))-\lambda^*\mathbbb{1} \Big) \geq 0   \quad \forall \, \beta \in K.  
\end{equation}
Due to the affine relation in \eqref{eq_flexibility_allocation_compact}, $(\beta^*+\mathbbb{1})$ and $(\beta^*-\mathbbb{1})$ 
are both in $K$. By substituting these two vectors in \eqref{eq_optimality2} independently, we get
\begin{equation}\label{eq_sum_zero}
\mathbbb{1}^\top(\sum_{n\in\mathcal{N}} \nabla_{x}(D_n(x^*_n))-\lambda^* \mathbbb{1} \Big) = 0.
\end{equation}
Note that for any $x\in \Theta$, we can always find $\beta \in K$ such that $x=A\beta+b$. Then, for any $x\in \Theta$, we have
 \begin{equation}
  \begin{aligned}
&(x-x^*)^\top \nabla_{x} (\sum_{n\in\mathcal{N}} D_n( x_n^*))\\ 
 =&(x-x^*)^\top\Big( \nabla_{x} (\sum_{n\in\mathcal{N}} D_n( x_n^*))-\lambda^*\mathbbb{1} \Big)\\
 =& (\beta-\beta^*)^\top  \Big( \nabla_{x} (\sum_{n\in\mathcal{N}} D_n( x_n^*))-\lambda^*\mathbbb{1} \Big)\\
 -& \frac{1}{N}(\beta-\beta^*)^\top\mathbbb{1}\mathbbb{1}^\top \Big( \nabla_{x} (\sum_{n\in\mathcal{N}} D_n( x_n^*))-\lambda^*\mathbbb{1} \Big) \\
 =& (\beta-\beta^*)^\top\Big( \nabla_{x} (\sum_{n\in\mathcal{N}} D_n( x_n^*))-\lambda^*\mathbbb{1} \Big)\geq 0
 \end{aligned}   
 \end{equation}
  where the first and third equalities come from \eqref{eq_flexibility_balance} and \eqref{eq_sum_zero} respectively. The above inequality implies $x^*$ is an optimizer of \eqref{eq_optimization2} and its uniqueness follows from the strict convexity of the objective function in \eqref{eq_optimization2}.
\end{proof}

\begin{proof}[Proof of Proposition~\ref{convergence}]   
First, we write our algorithm in a compact form. Bearing in mind \eqref{eq_price} and \eqref{eq_flexibility_allocation}, we can rewrite \eqref{eq_update_beta} and \eqref{eq_update_gamma} as
\begin{equation}
\begin{aligned}
&\tilde{\beta}_n^{k+1} = \beta_n^k - \rho_n\big (f_n(\beta_n^k,\beta_{-n}^k )-1/{N} \mathbbb{1}^{\top} \gamma^k + \gamma_n^k\big), 
\end{aligned}
\end{equation}
\begin{equation}
\begin{aligned}
    \gamma_n^{k+1} & = \proj_{\mathbb{R}^+} \big( \gamma_n^k  + \nu_n \big(2\mathbbb{1}^{\top}(\beta^k-\beta^{k+1})/N\\
     & +(2\beta_n^{k+1}-\beta_n^k)-\hat{x}_n+\xf/N\big) \big).
\end{aligned}
\end{equation} 
By letting  $R=\diag(\rho_n)_{n\in\mathcal{N}}$, $V=\diag(\nu_n)_{n\in\mathcal{N}}$ and noting that $\proj_{\Psi}(\tilde{\beta})=\argmin_{z\in \Psi} ||z-\tilde{\beta}||$, the above dynamics can be written as
\begin{equation}\label{eq_compact}
  \begin{aligned}
&\beta^{k+1}=\proj_{\Psi}\big(\beta^k-R(F(\beta^k)+A^\top\gamma^k)\big),\\
&\gamma^{k+1}=\proj_{\mathbb{R}^+_{N}}\big(\gamma^k+V(2A\beta^{k+1}-\hat{b})\big),
\end{aligned}  
\end{equation}
where $A= I- \frac{1}{N} \mathbbb{1}\mathbbb{1}^{\top}$ and $\hat{b}=\col(x_n-\xf/N)_{n\in\mathcal{N}}$. The above dynamics then have the similar structure as the preconditioned forward-backward algorithm in \cite{belgioioso2018projected}. Furthermore, the compact algorithm \eqref{eq_compact} is essentially the Banach-Picard iteration as,
\begin{equation}\label{eq_iteration}
    \omega^{k+1}=\mathcal{BP}(\omega^k)=(\mathrm{Id}+\Phi^{-1}\mathcal{B})^{-1}\circ(\mathrm{Id}-\Phi^{-1}\mathcal{A})\omega^k
\end{equation}
where $\mathrm{Id}$ is the identity mapping and $\omega=\col(\beta,\gamma)$. The mappings $\mathcal{A}$, $\mathcal{B}$ and the matrix $\Phi$ are defined as
\begin{equation}
   \mathcal{A}=\begin{bmatrix}
F(\beta)\\
\hat{b}
\end{bmatrix}, \mathcal{B}=\begin{bmatrix}
\N_{\Psi}(\beta)+A^{\top}\gamma\\
\N_{\mathbb{R}^+_N}(\gamma)-A\beta
\end{bmatrix},\Phi=\begin{bmatrix}
R^{-1}& -A^{\top}\\
- A & V^{-1}
\end{bmatrix}
\end{equation}
where $\N$ is the normal cone operator. Note that $\Psi$ is closed and convex, the mapping $F$ is $\eta_F$-monotone and $\kappa_F$-Lipschitz continuous by Lemma \ref{monolip}, and $\lambda_{\max}(A)=1$. 

Following \cite[Chapter 10.1]{facchinei2003finite}, we write the Karush–Kuhn–Tucker (KKT) condition for the VI problem in Proposition \ref{vGNE} as,
\begin{equation}
 \begin{aligned}
    & 0\in \N_{\Psi}(\beta^*)+F(\beta^*)+ A^{\top} \gamma^* \\
    & 0 \in \N_{\mathbb{R}^+_N}(\gamma^*)-(A\beta^*-\hat{b})
\end{aligned}   
\end{equation}
where the solution $\beta^*$ is the v-GNE of the game $G$ and the solution $\gamma^*$ is the dual variables. Note that the solution of this KKT condition is exactly the zero of the mapping $\mathcal{A}+\mathcal{B}$. Therefore, by \cite[Theorem~1]{belgioioso2018projected}, if the step sizes $\rho_n$ and $\nu_n$ satisfy \eqref{eq_step_sizes}, then the sequence $(\beta^k, \gamma^k)$ generated by Algorithm \ref{alg1} converges to $(\beta^*, \gamma^*)$.

Finally, we investigate the convergence rate of the iteration \eqref{eq_iteration} to the fix point $\omega^*=\col(\beta^*, \gamma^*)$. According to \cite[Lemma~5]{belgioioso2018projected}, if the step sizes $\rho_n$ and $\nu_n$ satisfy \eqref{eq_step_sizes}, the iteration \eqref{eq_iteration} is $\xi$-averaged with $\xi\in (0,1)$, that is
\begin{equation}
\begin{aligned}
    &\|\mathcal{BP}(\omega^k)-\mathcal{BP}(\omega^*)\|^2  \leq  \|\omega^k-\omega^*\|^2 \\  -&\frac{1-\xi}{\xi}  \|(\mathrm{Id}-\mathcal{BP})(\omega^k)-(\mathrm{Id}-\mathcal{BP})(\omega^*)\|^2.  
\end{aligned}
\end{equation}
Due to the fact that $\omega^{k+1}=\mathcal{BP}(\omega^k)$ and $\omega^*=\mathcal{BP}(\omega^*)$, the above inequality implies
\begin{equation}\label{eq_average}
    \|\omega^{k+1}-\omega^*\|^2  \leq  \|\omega^k-\omega^*\|^2   -\frac{1-\xi}{\xi}  \|\omega^{k+1}-\omega^k\|^2,  
\end{equation}
and $\|\omega^{k+1}-\omega^*\|^2  \leq  \|\omega^k-\omega^*\|^2$, which means $\omega^k$ sublinearly converges. More specifically, we sum up \eqref{eq_average} from $k=0$ to any $k>0$, we obtain, 
\begin{equation}
    \frac{1}{k} \sum^k_{j=0}\|\omega^{j+1}-\omega^j\|^2  \leq  \frac{\xi}{k(1-\xi)}( \|\omega^0-\omega^*\|^2 -\|\omega^{k+1}-\omega^*\|^2),
\end{equation}
which implies
\begin{equation}
    \min_{j=0,1,...,k} \|\omega^{j+1}-\omega^j\|^2  \leq  \frac{\xi}{k(1-\xi)}\|\omega^0-\omega^*\|^2=\mathcal{O}(1/k).
\end{equation}
Overall, we conclude that the sequence $(\beta^k, \gamma^k)$ sublinearly converges to $(\beta^*, \gamma^*)$ with convergence rate in \eqref{eq_rate}.
\end{proof} 

\bibliographystyle{IEEEtran}
\bibliography{ref}

\end{document}